\documentclass[letterpaper]{article}
\pdfoutput=1
\usepackage[T1]{fontenc}
\usepackage[utf8]{inputenc}
\usepackage{csquotes}
\usepackage[english]{babel}

\usepackage{microtype}

\usepackage{bbm} 

\usepackage{bm}

\usepackage{graphicx}

\usepackage{amssymb,amsmath,amsthm}

\usepackage{color}
\usepackage[dvipsnames]{xcolor}

\usepackage[%
    backend=biber,
    style=alphabetic,
    sorting=nyt,
    isbn=true,
    maxbibnames=10,
    mincrossrefs=99 
]{biblatex}
\usepackage{authblk}

\DeclareSourcemap{%
	\maps[datatype=bibtex]{%
		\map[overwrite]{%
			\step[fieldsource=doi, final]
			\step[fieldset=url, null]
			\step[fieldset=eprint, null]
		}
	}
}

\usepackage{amssymb,amsmath,amsthm}

\usepackage{mathtools}
\mathtoolsset{centercolon}
\DeclareMathOperator{\poly}{poly}

\DeclareMathOperator{\rt}{rt}
\DeclareMathOperator{\rs}{rs}
\DeclareMathOperator{\tr}{Tr}

\DeclareMathOperator{\dom}{domain}
\DeclareMathOperator{\img}{image}
\DeclareMathOperator{\Ex}{\mathbf{E}}
\DeclareMathOperator{\Prob}{P}

\DeclarePairedDelimiter\ket{\vert}{\rangle}

\DeclarePairedDelimiterX\ketbra[2]{\vert}{\vert}{#1 {\delimsize\rangle\langle} #2}
\DeclarePairedDelimiter\set{\{}{\}}
\DeclarePairedDelimiter\abs{\lvert}{\rvert}

\DeclarePairedDelimiterXPP\bigo[1]{O}{(}{)}{}{#1}
\DeclarePairedDelimiterXPP\littleo[1]{o}{(}{)}{}{#1}
\DeclarePairedDelimiterXPP\bigtildeo[1]{\tilde O}{(}{)}{}{#1}
\DeclarePairedDelimiterXPP\bigomega[1]{\Omega}{(}{)}{}{#1}
\DeclarePairedDelimiterXPP\bigtheta[1]{\Theta}{(}{)}{}{#1}
\DeclarePairedDelimiterXPP\diagonal[1]{\diag}{(}{)}{}{#1}
\DeclarePairedDelimiterXPP\rnumber[1]{\rt}{(}{)}{}{#1}
\DeclarePairedDelimiterXPP\rsize[1]{\rs}{(}{)}{}{#1}
\DeclarePairedDelimiterXPP\hierprod[2]{Π_{#1}}{(}{)}{}{#2}
\DeclarePairedDelimiterXPP\domain[1]{\dom}{(}{)}{}{#1}
\DeclarePairedDelimiterXPP\image[1]{\img}{(}{)}{}{#1}
\DeclarePairedDelimiterXPP\trace[1]{\tr}{[}{]}{}{#1}
\DeclarePairedDelimiterXPP\ptrace[2]{\tr_{#1}}{[}{]}{}{#2}
\DeclarePairedDelimiterXPP\probability[1]{\Prob}{[}{]}{}{#1}
\DeclarePairedDelimiterXPP\expectation[1]{\Ex}{[}{]}{}{#1}

\makeatletter
\DeclareRobustCommand{\rvdots}{%
	\vbox{
	  	\baselineskip4\p@\lineskiplimit\z@
	  	\kern-\p@
	  	\hbox{.}\hbox{.}\hbox{.}
}}
\makeatother
\makeatletter
\def\Ddots{\mathinner{\mkern1mu\raise\p@
\vbox{\kern7\p@\hbox{.}}\mkern2mu
\raise4\p@\hbox{.}\mkern2mu\raise7\p@\hbox{.}\mkern1mu}}
\makeatother

\newcommand*{\idm}{\ensuremath{\mathbbm{1}}}
\renewcommand{\vec}[1]{\ensuremath{\mathbf{#1}}}

\usepackage{booktabs}

\usepackage[lined,algoruled,algosection,linesnumbered,noend]{algorithm2e}
\SetAlCapSkip{0.5em}
\SetEndCharOfAlgoLine{} 
\SetKwInOut{Data}{Data}
\SetKwInOut{Input}{Input}
\SetKwInOut{Output}{Output}
\SetKwFor{For}{for}{\string:}{endfor}
\SetKwFor{While}{while}{\string:}{endw}
\SetKwFor{ForEach}{foreach}{\string:}{endfch}

\SetFuncArgSty{}
\SetArgSty{}

\SetKwProg{Fn}{function}{\string:}{endfunction}
\SetKwFunction{execute}{execute}

\usepackage[%
	labelfont=bf
]{caption}
\usepackage{subcaption}

\definecolor{dark-red}{rgb}{0.4,0.15,0.15}
\definecolor{dark-blue}{rgb}{0.15,0.15,0.4}
\definecolor{medium-blue}{rgb}{0,0,0.5}

\definecolor{mycomment}{rgb}{0.3,0.7,0.8}
\definecolor{mygray}{rgb}{0.5,0.5,0.5}
\definecolor{lightgray}{rgb}{0.95,0.95,0.95}
\definecolor{mymauve}{rgb}{0.58,0,0.82}

\usepackage[unicode=true]{hyperref}
\hypersetup{%
	breaklinks=true,
	colorlinks,
	linkcolor={dark-blue},
	citecolor={dark-blue},
	urlcolor={medium-blue},
	pdfauthor = {Eddie Schoute},
	pdfpagemode=UseOutlines,
	pdfborder={0 0 0},
	pdfcreator={},
	pdfproducer={}
}

\usepackage{enumitem}
\newlist{ienumerate}{enumerate*}{1}
\setlist*[ienumerate,1]{%
    label=(\roman*),
}

\usepackage{pgfplots}
\pgfplotsset{compat=newest}
\usepgfplotslibrary{fillbetween}
\usepackage{tikz}
\usetikzlibrary{shapes.geometric}
\usetikzlibrary{backgrounds}
\usepackage{tikzscale}

\usepackage[nameinlink,capitalise]{cleveref}
\crefname{figure}{Figure}{Figures}
\crefname{equation}{}{} 
\Crefname{equation}{Equation}{Equations}

\newtheorem{theorem}{Theorem}[section]
\newtheorem{lemma}[theorem]{Lemma}
\newtheorem{corollary}[theorem]{Corollary}
\theoremstyle{definition}
\newtheorem{definition}[theorem]{Definition}

\newcommand{\cnot}[0]{{\textsc{cnot}}}

\newcommand{\cknot}[1]{\ensuremath{\textsc{c}^{#1}\kern-0.1em\textsc{not}}}

\newcommand{\sw}[0]{{\textsc{swap}}} 

\newcommand{\MaxFlow}[0]{{\textsc{Max Flow}}}

\title{Surface code compilation via edge-disjoint paths }

\author[1]{Michael Beverland}
\author[1]{Vadym Kliuchnikov}
\author[2,3,4]{Eddie Schoute\thanks{\href{mailto:eschoute@lanl.gov}{eschoute@lanl.gov}}}
\affil[1]{Microsoft Quantum}
\affil[2]{Joint Center for Quantum Information and Computer Science, University of Maryland}
\affil[3]{Institute for Advanced Computer Studies, University of Maryland}
\affil[4]{Department of Computer Science, University of Maryland}

\begin{document}
\maketitle

\begin{abstract}
	We provide an efficient algorithm to compile quantum circuits for fault-tolerant execution.
	We target surface codes, which form a 2D grid of logical qubits with nearest-neighbor logical operations.
	Embedding an input circuit's qubits in surface codes can result in long-range two-qubit operations across the grid.
	We show how to prepare many long-range Bell pairs on qubits connected by edge-disjoint paths of ancillas in constant depth that can be used to perform these long-range operations.
	This forms one core part of our \emph{Edge-Disjoint Paths Compilation} (EDPC) algorithm,
	by easily performing many parallel long-range Clifford operations in constant depth.
	It also allows us to establish a connection between surface code compilation and several well-studied edge-disjoint paths problems.
	Similar techniques allow us to perform non-Clifford single-qubit rotations
	far from magic state distillation factories.
	In this case, we can easily find the maximum set of paths by a max-flow reduction, which forms the other major part of EDPC.
	EDPC has the best asymptotic worst-case performance guarantees on the circuit depth for compiling
	parallel operations when compared to related compilation methods based on \sw{}s and network coding.
	EDPC also shows a quadratic depth improvement over sequential Pauli-based compilation
	for parallel rotations requiring magic resources.
	We implement EDPC and find significantly improved performance for circuits built from parallel \cnot{}s,
    and for circuits which implement the multi-controlled $X$ gate \cknot{k}.
\end{abstract}


\clearpage

\tableofcontents

\clearpage

\section{Introduction}
\label{sec:intro}
Quantum hardware will always be somewhat faulty and subject to decoherence, due to inevitable fabrication imperfections and the impossibility of completely isolating physical systems.
For large computations it becomes a certainty that faults will occur among the many qubits and operations involved.
\emph{Fault-tolerant quantum computation} (FTQC) can be implemented despite this by encoding the information in a quantum error correcting code and applying logical operations which are carefully designed to process the encoded information with an acceptably low effective error rate.

The surface code~\cite{Kitaev2003,Bravyi1998} provides a promising approach to implement FTQC.
Firstly, it can be implemented using geometrically local operations on a patch of qubits in a 2D grid,  which is the natural setting for many hardware platforms including superconducting~\cite{Fowler2012,Chamberland2020b} and Majorana~\cite{Karzig2017} qubits.
Secondly, the logical qubits it encodes remain protected even for relatively high noise rates, with a threshold of around 1\%~\cite{Wang2011}.
Thirdly, a sufficiently general set of elementary logical operations can be performed fault tolerantly on qubits encoded in the surface code using \emph{lattice surgery}~\cite{Horsman2012}.  
By tiling the plane with surface code patches, a 2D grid of logical qubits is formed, where the elementary operations are geometrically local; see \cref{fig:logical-layout}.
When combined with magic state distillation~\cite{Bravyi2005} these operations become universal for quantum computing.
Indeed this approach, which we will refer to as the \emph{surface code architecture}, is seen as among the most promising by many research groups and companies working in quantum computing~\cite{Fowler2012,Chao2020,YoderIBM2017,Fruchtman2016}.

\begin{figure}[h]
	\centering
	\includegraphics[width=0.85\textwidth]{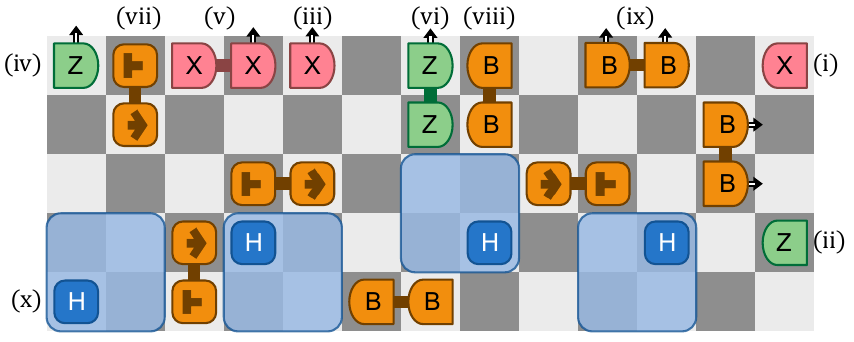}
	\caption{%
		Logical qubits (light and dark gray patches) encoded in the surface code form a 2D grid. 
		The elementary operations can be applied on any lattice translations of those shown. 
		Their times in units of surface code logical time steps are as follows. 
		$\mathit 0$ \textit{logical time steps:} Single-qubit preparation in the $X$ basis (i), and the $Z$ basis	(ii).
		Destructive single-qubit measurement, which moves the patch outside of the code space,
		in the $X$ basis (iii), and the $Z$ basis (iv) take $0$ steps.
		$\mathit 1$ \textit{logical time step:} Two-qubit measurement of $XX$ (v) and $ZZ$ (vi).
		A move of a logical qubit from one patch to an unused patch (vii).	
		Two-qubit preparation (viii) and destructive measurement (ix) in the Bell basis.
		$\mathit 3$ \textit{logical time steps:} A Hadamard gate, which uses three ancilla patches (x). See \cref{sec:surfaceCode} for further details.
	}\label{fig:logical-layout}
\end{figure}

In this work, we seek to minimize the resources required to fault-tolerantly implement a quantum algorithm using the surface code architecture, which we will refer to as the \textit{surface code compilation problem}.
For concreteness, we will assume that the input quantum algorithm is expressed as a quantum circuit composed of preparations and destructive measurements of individual qubits in the $Z$ or $X$ basis, controlled-not (\cnot{}), Pauli-$X$, -$Y$, and -$Z$, Hadamard ($H$), Phase ($S$) and $T$ gates.
Our results can be easily generalized to broader classes of input quantum circuits.
The output is the quantum algorithm executed using the elementary logical surface code operations shown in \cref{fig:logical-layout}.
Ultimately, we would like to minimize the \emph{physical space-time cost},
which is the product of the number of physical qubits and the time required to run an algorithm.
To avoid implementation details, we instead minimize the more abstract \emph{logical space-time cost}, which is the number of logical qubits (the circuit width) multiplied by the number of logical time steps (the circuit depth) of the algorithm expressed in elementary surface code operations.
The logical and physical space-time costs are expected to be 1-to-1 and monotonically related (see \cref{sec:logical-physical-cost}), such that minimizing the former should minimize the latter. 

\begin{figure}
	\begin{subfigure}[b]{0.5\textwidth}
		\centering
		\includegraphics[width=\textwidth]{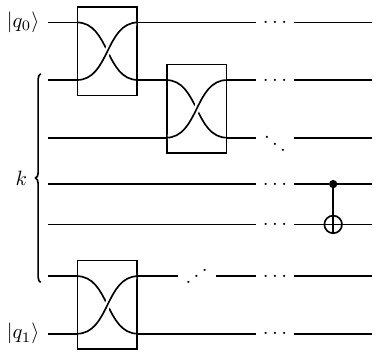}
		\caption{With \sw{} operations}\label{fig:compareSwap}
	\end{subfigure}%
	\begin{subfigure}[b]{0.5\textwidth}
		\centering
		\includegraphics[scale=0.75]{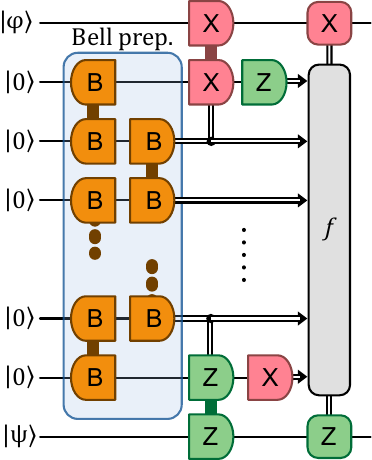}
		\caption{With Bell pair}\label{fig:compareTeleportation}
	\end{subfigure}
	\caption{%
		Application of a \cnot{$(q_0,q_1)$} on qubits at distance $k+1$ using surface code operations in two ways.
		(a) Using a \sw{}-based approach requires $\bigomega{n}$ depth using operations from \cref{fig:logical-layout},
		while (b) generating and consuming a Bell pair~\cite{Litinski2017} can be implemented in constant depth.
		The classical function $f$ computes Pauli corrections on the output qubits.
        }\label{fig:compareSwapTeleportation}
\end{figure}

A well-established approach to implement surface code compilation is known as sequential Pauli-based computation~\cite{Litinski2019}, where non-Clifford operations are implemented by injection using Pauli measurements, and Clifford operations are conjugated through the circuit until the end.
The circuit that is run in this approach then consists of a sequence of high-weight Pauli measurements which can have overlapping support leading them to be measured one after the other.
For large input circuits this can be problematic because highly parallel input circuits can become serialized with prohibitive runtimes. 

A major challenge to solve the surface code compilation problem is that quantum algorithms typically involve operations between logical qubits that are far apart when laid out in a 2D grid.
One approach to deal with a long-range gate is to swap logical qubits around until the pair of interacting qubits are next to one-another~\cite{Childs2019}.
However, this can result in a deep circuit, see \cref{fig:compareSwap}.
A more efficient approach is to create long-range entanglement by producing Bell pairs, which for example can be used to implement a long-range \cnot{} with a constant-depth circuit~\cite{Briegel1998,Litinski2017,JavadiAbhari2017} (see \cref{fig:compareTeleportation}).
Both of these approaches can be implemented with the elementary operations of the surface code.

Moreover, algorithms typically consist of many long-range operations that can ideally be performed in parallel.
For swap-based approaches, this can be done by considering a permutation of the logical qubits which is implemented by a sequence of swaps~\cite{Lao2018,Murali2019,Zulehner2019}. 
Finding these \sw{} circuits reduces to a routing problem on graphs~\cite{Childs2019,Steiger2018}.
There are efficient algorithms that solve this problem for certain families of graphs~\cite{Alon1994,Childs2019},
but finding a minimal depth solution is NP-hard in general~\cite{Banerjee2017}.
Alternatively, \emph{linear network coding} can be used to prepare many long-range pairs in constant depth~\cite{Leung2010,Kobayashi2009,Kobayashi2011,Satoh2012,Hahn2019,Beaudrap2020a},
and then these Bell pairs can be used to implement operations on pairs of distant qubits.
But a major barrier for using linear network coding is the lack of known efficient algorithms to find linear network codes.

In this paper, we provide a solution to the surface code compilation problem which generalizes
the use of entanglement for long-range \cnot{}s discussed above to the implementation of many long-range operations in parallel. 
In particular, we propose the Edge-Disjoint Paths Compilation (EDPC) algorithm,
which is a computationally efficient classical algorithm tailored to the elementary operations of the surface code architecture.
We find evidence that our EDPC algorithm significantly outperforms other
approaches by performing a detailed cost analysis for the execution of a set of quantum circuits benchmarks.

EDPC reduces the problem of executing quantum circuits to problems in graph theory.
Logical qubits correspond to graph vertices,
and there is an edge between qubits if elementary surface code operations can be applied between them.
We show how to perform multiple long-range \cnot{}s in constant depth
along a set of edge-disjoint paths (EDP) in the graph. 
In other words, long-range \cnot{}s can be performed simultaneously, in one round, if their controls and targets are connected by edge-disjoint paths.
This leads to the well-studied problem of finding maximum EDP sets~\cite{Kleinberg1996}.
The ability to perform long-range \cnot{}s
along with the elementary operations allows compilation of Clifford operations.
We also give a construction for EDP sets that are asymptotically optimal in the depth
of worst-case sets of independent \cnot{}s.

\begin{table}
    \centering
    \begin{tabular}{lrrr}\toprule
        & \multicolumn{3}{c}{Input circuit (compiled depth)} \\ \cmidrule(l){2-4}
      Algorithm & 1 \cnot{} & $n/2$ parallel \cnot{}s  & $k$ parallel rotations \\ \midrule
      Sequential Pauli & 0 & 0 & $\bigtheta{k}$ \\
      \sw{} & $\bigtheta{\sqrt n}$ & $\bigtheta{\sqrt n}$ & $\bigtheta{\sqrt n}$ \\
      Network coding & $\bigtheta{1}$ & $\bigomega{\sqrt n}$ & $\bigomega{\sqrt k}$ \\
      \textbf{EDPC} & $\bm{{\bigtheta{1}}}$ & $\bm{{\bigtheta{\sqrt n}}}$ & $\bm{{\bigtheta{\sqrt k}}}$ \\
      \bottomrule%
    \end{tabular}
    \caption{%
        A comparison in the depth of surface code compilation algorithms
        (that use $\bigtheta{n}$ space) for various input circuits of width $n$.
        We compare the worst-case performance for a single long-range \cnot{} gate,
        for \cnot{} circuits with $n/2$ parallel \cnot{} gates,
        and for $k$ rotations, with $k \in \mathbb N$, that need to be performed at the boundary.
    }\label{tab:comparison}
\end{table}

The final operations that complete our gate set for universal quantum computation with the surface code are $T$ gates.
The $T$ gates are not natural operations on the surface code, but can be implemented fault-tolerantly by consuming specialized resource states, called \emph{magic states}. 
Magic states can be produced using a highly-optimized process called magic state distillation,
which we assume occurs independently of the computation on our code.
We assume that logical magic states are available in a specified region of the grid.
EDPC reduces magic state delivery to simple \MaxFlow{} instances that have known efficient algorithm~\cite{Ford1956}.
We compare the depth of input circuits compiled using surface code compilation algorithms
in the literature and EDPC in \cref{tab:comparison}.

The outline of the paper is as follows.
In \cref{sec:circuit-components}, we construct key higher-level components from the basic surface code operations in \cref{fig:logical-layout} including simple long-range operations.
These long-range operations allow us to perform many parallel \cnot{} operations
given vertex-disjoint and edge-disjoint paths that connect the data qubits in \cref{sec:edge-disjoint-path-algos}.
Because of its importance to the algorithms,
there we also compare the state of the art graph algorithms for finding vertex-disjoint or edge-disjoint sets of paths
and analyze their relation to our algorithms.
We complete our gate set by giving an algorithm for efficient remote rotations using magic states at the boundary in \cref{sec:magicState}.
Putting parallel long-range \cnot{} and remote rotations together,
we construct our circuit compilation algorithm, EDPC, in \cref{sec:EDPCalogrithm}.
Finally, we compare the performance of EDPC to prior surface code compilation work in \cref{sec:results},
note its connections to network coding,
and give numerical results comparing the space-time performance with a \sw{}-based compilation algorithm.

\section{Key circuit components from surface code operations}\label{sec:circuit-components}
Recall that our goal in this work is to develop an efficient classical compilation algorithm which re-expresses a quantum algorithm into one that uses the elementary operations of the surface code with a low logical space-time cost.
In \cref{sec:surfaceCode} we give an overview of the surface code
and justify the resource costs of the elementary operations shown in \cref{fig:logical-layout}.
The initial quantum algorithm is assumed to be expressed as a circuit diagram involving preparations and measurements of individual qubits in the computational basis, controlled-not (\cnot{}), Pauli-$X$, -$Y$, and -$Z$, Hadamard ($H$), Phase ($S$) and $T$ gates.
In this section we build and calculate the cost of some key circuit components from the elementary surface code operations in \cref{fig:logical-layout}.
The contents of this section are reproductions or straightforward extensions of previously-known circuits.

\subsection{Single-qubit operations}\label{sec:natural-gates}
Some of the operations of the input circuit can be implemented directly with elementary surface code operations, namely the preparation and measurement of individual qubits in the measurement basis, and the Hadamard gate (provided three neighboring ancillary patches are available as ancillas, see \cref{fig:logical-layout}).
Pauli operations do not need to be implemented at all since 
they can be commuted through Clifford gates and arbitrary Pauli gates~\cite{Knill2005}
and can therefore be tracked classically and merged with the final measurements.
For this reason, while we occasionally explicitly provide the Pauli corrections where instructive, we often show equivalence of two circuits only up to Pauli corrections.
The remaining single-qubit operations in the input circuit, namely the $S$ and $T$ gates, can be implemented using magic states and is addressed in \cref{sec:magicState}.

\subsection{Local \cnot{} and \sw{} gates}
An important circuit component is the \cnot{} gate, which can be implemented as shown in \cref{fig:cnotGate}~\cite{Zilberberg2008}. 
The qubits involved in this example are stored in adjacent patches,
i.e., it is local.
Another useful operation is a \sw{} of a pair of qubits stored in nearby patches. 
The surface code's move operation shown in \cref{fig:logical-layout}
gives a straightforward way to implement this as shown in \cref{fig:swap-via-move}.
With these implementations, the \cnot{} requires one ancilla patch, while \sw{} requires two.
Both are depth $2$.

\begin{figure}[h!]
	\centering
    \begin{subfigure}[t]{0.4\textwidth}
        \centering
       	\includegraphics[scale=0.6]{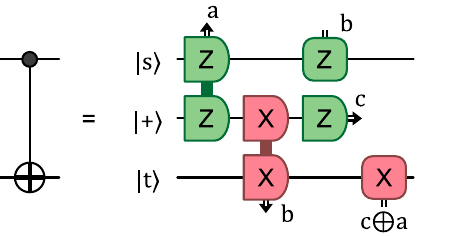}
        \caption{\cnot{} gate}
        \label{fig:cnotGate}
    \end{subfigure}\hfill%
    \begin{subfigure}[t]{0.59\textwidth}
        \centering
        \includegraphics[scale=0.6]{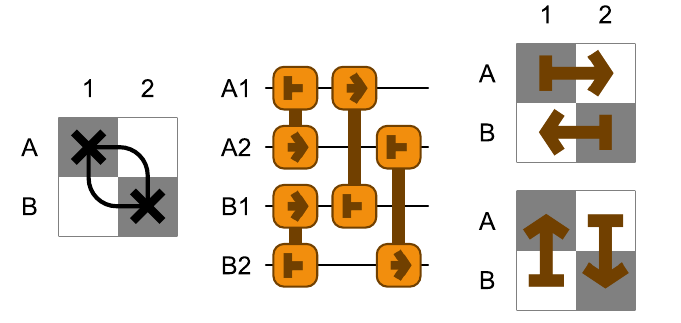}
        \caption{\sw{} gate}
        \label{fig:swap-via-move}
    \end{subfigure}
	\caption{%
		A \cnot{} gate can be implemented in depth $2$ using $ZZ$ and $XX$ joint measurements
		with a $\ket +$ ancilla state,
		followed by classically controlled Pauli corrections.
		The \sw{} gate can be implemented using four move operations and two ancillas in depth $2$.
	}
\end{figure}

\subsection{Long-range \cnot{} using \sw{} gates}\label{sec:longRangeCnotSwap}
\begin{figure}[h]
		\centering
		\includegraphics[width=\textwidth]{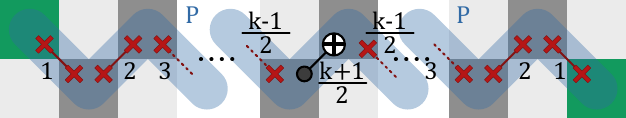}
		\caption{%
		A non-local \cnot{} can be implemented using \sw{}s which takes depth $2\lceil\frac{k-1}{2}\rceil$ using a zig-zag of ancilla patches along the path $P$ of length $k$.
		The figure shows the case when $k$ is odd and \sw{}s depth is $2(k-1)$.
		The patches on the path can store other logical information, which will simply be moved during the \sw{} gates. 
		The patches adjacent to the path are ancillas which are used to implemented the \sw{} gates.
		}\label{fig:long-range-cnot-swap}
\end{figure}

Typical input circuits for surface code compilation will involve \cnot{} operations on pairs of qubits that are far apart after layout.
A very intuitive approach to apply a long-range \cnot{$(q_1, q_2)$} gate is shown in \cref{fig:long-range-cnot-swap}.
This involves making use of \sw{} gates to first move the qubits $q_1$ and $q_2$ so that they are near one another, and then use the local \cnot{} gate in \cref{fig:cnotGate}.
Let the path $P=v_1 v_2 \dots v_k$, for $k \in \mathbb N$, where $v_1 = q_1$ and $v_k = q_2$.
As each swap has depth $2$, we get a circuit of depth $2\lceil\frac{k-1}{2}\rceil$ 
since we can perform \sw{}s on either end simultaneously.
Afterwards, the two qubits are adjacent and we simply perform a \cnot{} in depth $2$.

A lower bound on the depth it takes to perform a long-range \cnot{} gate using swaps
is proportional to the length of the shortest $q_1$-$q_2$ path.
To move a qubit $k$ patches using \sw{}s takes depth exactly $2k$.
Therefore, to move control and target to the middle of the shortest path connecting them,
it must take time proportional to at least half the length of the path.

\subsection{Long-range \cnot{} using a Bell pair}\label{sec:longRangeCnot}
\begin{figure}[h]
	\centering
    \begin{subfigure}[t]{0.47\textwidth}
        \centering
        \includegraphics[scale=0.7]{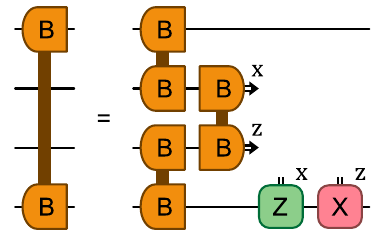}
        \caption{Preparing a longer-range Bell pair}
        \label{fig:longer-range-bell}
    \end{subfigure}\hfill%
    \begin{subfigure}[t]{0.47\textwidth}
        \centering
        \includegraphics[scale=0.55]{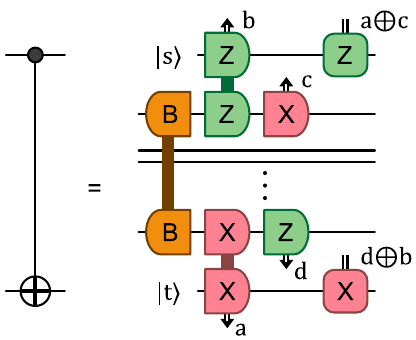}
        \caption{\cnot{} by consuming a Bell pair}\label{fig:cnot-via-bell}
    \end{subfigure}
    \par\bigskip
	\begin{subfigure}[t]{0.47\textwidth}
		\centering
		\includegraphics[width=0.9\textwidth]{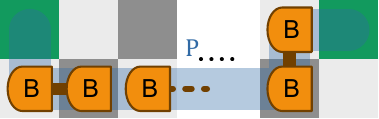}
		\caption{Preparing and consuming a Bell pair for long-range \cnot{} I}\label{fig:nonlocal-cnot-bell-1}
	\end{subfigure}\hfill%
	\begin{subfigure}[t]{0.47\textwidth}
		\centering
		\includegraphics[width=0.9\textwidth]{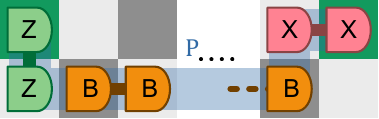}
		\caption{Preparing and consuming a Bell pair for long-range \cnot{} II}\label{fig:nonlocal-cnot-bell-2}
	\end{subfigure}
	\caption{%
	A long-range \cnot{} can be implemented in depth $2$ by first preparing a Bell pair. 
	(a) Joining Bell pairs with Bell measurements. This can be iterated to form a long-range Bell pair along any path of ancillas in depth $2$. 
	(b) A Bell pair can be used to apply a \cnot{}.
	(c,d) The first and second steps of a depth-2 circuit that implements a \cnot{} between a pair of patches at the end of a path of ancilla patches by preparing and consuming a Bell pair. 
	}\label{fig:longRangeCnot}
\end{figure}
A circuit component that we make extensive use of in this paper is the long-range \cnot{} using a Bell pair~\cite{Litinski2017}.
This allows us to apply \cnot{}s in depth $2$ between any pair of qubits (provided there is a path of ancilla qubits which connects them).

To understand the construction, we first show in \cref{fig:longer-range-bell} how to prepare a longer-range Bell pair from two Bell pairs.
By iterating this construction one can form a circuit to prepare a long-range Bell pair at the ends of any path of adjacent ancilla patches in depth $2$.
Next, we show in \cref{fig:cnot-via-bell} how to implement a \cnot{} operation between qubits stored in patches neighboring a pair of patches storing a Bell pair.
Putting these together, using a path of ancilla patches between a pair of qubits,
a long-range \cnot{} can be implemented in depth $2$ in a two-step circuit shown in \cref{fig:nonlocal-cnot-bell-1} and \cref{fig:nonlocal-cnot-bell-2} respectively.
This approach can be used to implement the \cnot{} in depth $2$ circuit using any path from the control to the target qubit which starts with a vertical edge and ends with a horizontal edge.
There is also flexibility in the precise arrangement of the Bell pairs and Bell measurements along the path using the circuits in \cref{sec:cnots-via-bell}.

Note that here we have focused on implementing a long-range \cnot{} by constructing and consuming a Bell pair.
However a similar strategy (of first preparing a long-range Bell pair in the patches at the ends of a path of ancillas) can be used to implement other long-range operations such as teleportation.

\section{Parallel long-range \cnot{}s using Bell pairs}\label{sec:edge-disjoint-path-algos}

Here, we generalize the use of Bell pairs from the setting of compiling an individual non-local \cnot{} gate into surface code operations to the setting in which a set of parallel non-local \cnot{} gates are compiled.
In \cref{fig:logical-layout} and the circuit components in \cref{sec:circuit-components}, ancilla qubits are used to perform some operations on data qubits.
To consider the compilation on large sets of qubits,
we must specify the location of data and ancilla qubits:
here we assume a 1 data to 3 ancilla qubit ratio,
as illustrated in \cref{fig:edpSchedule}.

In \cref{sec:edp-background} we discuss some relevant background on sets of vertex-disjoint paths (VDP) and sets of edge-disjoint paths (EDP) in graphs.
Then in \cref{sec:long-range-cnot-bell} we define the \emph{VDP subroutine} and the \emph{EDP subroutine}
that apply parallel \cnot{} gates at the ends of a particular type of VDP or EDP set.
In \cref{sec:EDP-analysis},
we show how to use the EDP subroutine to compile more general \cnot{} circuits
and prove bounds on the performance of this approach.

\subsection{Vertex-disjoint paths (VDP) and edge-disjoint paths (EDP)}\label{sec:edp-background}
In \cref{sec:longRangeCnot} we saw that a long-range \cnot{} could be implemented with the use of a Bell pair produced with a path of ancilla qubits connecting the control and target of the \cnot{}.
A barrier to implement multiple \cnot{}s simultaneously can arise when an ancilla resides in the paths associated with multiple different \cnot{}s. 
This motivates us to review some relevant theoretical background concerning sets of paths on graphs.

Given a graph $G$, a set of paths $\mathcal P$ is said to be a \textit{vertex-disjoint-path} (VDP) set if no pair of paths in $\mathcal P$ share a vertex, and an \textit{edge-disjoint-path} (EDP) set if no pair of paths in $\mathcal P$ share an edge.
Note that a set of vertex-disjoint paths is also edge-disjoint.
Further consider a set of \emph{terminal pairs} $\mathcal T = \set{(s_1, t_1), \dots, (s_k, t_k)}$
for terminals $s_i, t_i \in V(G)$, the vertices of $G$, and $i \in [k]$.
We then say that a set of paths $\mathcal P$ is a \emph{VDP set for} $\mathcal T$ (respectively an \emph{EDP set for} $\mathcal T$) if $\mathcal P$ is a VDP set (respectively an EDP set), and each path in $\mathcal P$ connects a distinct pair in $\mathcal T$.
These path sets do not necessarily connect all pairs in $\mathcal T$.
In what follows, we pay special attention to the square grid graph (see \cref{fig:graphs-grid}).
The grid graph is relevant for qubits in the surface code as shown in \cref{fig:logical-layout},
where the vertices correspond to code patches and edge connect vertices associated with adjacent patchs\footnote{%
	Later we will consider a modification of the square grid graph because our algorithms require
	some further restrictions on the paths, for example preventing them from passing through those vertices associated with data qubits. 
	It is unclear if all of the results in this section also apply for these modified graphs.
}.

The problems of finding a maximum (cardinality) VDP set for $\mathcal T$ or a maximum EDP set for $\mathcal T$ have been well-studied and there are known efficient algorithms capable of finding approximate solutions to each.
Unfortunately, on grids it is particularly hard to approximate the maximum VDP set.
In particular, for $N \coloneqq \abs{V(G)}$ there exist terminal sets for which no efficient algorithm can find an approximate solution
to within a $2^{\bigo{\log^{1-\epsilon} N}}$ factor of the maximum set size for any $\epsilon > 0$,
unless $\mathrm{NP} \subseteq \mathrm{RTIME}(N^{\poly \log N})$~\cite{Chuzhoy2018}.
However, efficient algorithms are available if one is willing to accept a looser approximation to the optimal solution.
For example, a simple greedy algorithm is an $\bigo{\sqrt N}$-approximation algorithm for finding the maximum VDP set~\cite{Kolliopoulos2004,Kleinberg2006:Approximation},
i.e., it produces a VDP set to within an $\bigo{\sqrt N}$ multiplicative factor of the optimal solution for any graph, not just the grid.
For grids, the best efficient algorithm that is known is an $\bigtildeo{N^{1/4}}$-approximation algorithm~\cite{Chuzhoy2015},
where $\bigtildeo{\cdot}$ hides logarithmic factors of $\bigo{\cdot}$.

The situation is better for approximation algorithms of the maximum EDP set:
There is a $\Theta(\sqrt N)$-approximation algorithm~\cite{Chekuri2006} for any graph,
and on grids \textcite{Aumann1995} showed an $\bigo{\log N}$-approximation algorithm that
was later improved to an $\bigo{1}$-approximation algorithm~\cite{Kleinberg1995,Kleinberg1996}.
In practice, these algorithms can be technical to implement and can have large constant prefactors in their solutions
that can be prohibitive for the instance sizes that we consider.
A simple greedy algorithm forms a $\bigo{\sqrt{N}}$-approximation algorithm~\cite{Kolliopoulos2004}
for finding a maximum EDP set on the two-dimensional grid
and does not suffer from the constant prefactors of the asymptotically superior alternatives.
The dominant runtime complexity of this greedy algorithm is mainly in finding shortest paths for each terminal pair,
giving a $\bigo{\abs{\mathcal T} N \log N}$,
runtime upper bound by Dijkstra's algorithm\footnote{%
It may be possible to improve the runtime by using a decremental dynamic all-pair shortest path algorithm;
it may be quicker to maintain a data structure for all shortest paths that can quickly be updated when edges are removed.
}.

It is informative to consider the comparative size of the maximum EDP and VDP sets for the same terminal set $\mathcal T$.
Since any VDP set is also an EDP set, the size of the maximum VDP set for $\mathcal T$ cannot be larger than the maximum EDP set for $\mathcal T$.
Moreover, one can construct some cases of $\mathcal{T}$ on the grid~\cite{Kleinberg1996} in which the maximum EDP set is a factor $\sqrt N$ larger than the maximum VDP set~\cite{Kleinberg1996}.
For example, consider the set of terminal pairs $\mathcal T = \set{((i,1), (L,i)) \mid i \in [L]}$ of an $L \times L$ grid graph,
where vertex $(i,j)$ denotes the vertex in row $i$ and column $j$.
All terminals can be connected by edge-disjoint paths but the maximum VDP set is of size one.

In \cref{sec:long-range-cnot-bell}, we show that both VDP and EDP sets for $\mathcal T$
can be used to form constant-depth compilation subroutines for disjoint \cnot{} circuits.
Ultimately, as will become clear in \cref{sec:long-range-cnot-bell}, each path in the EDP or VDP sets for $\mathcal T$ allows
us to implement one more \cnot{} gate in parallel by a compilation subroutine.
In this work, we focus on EDPs rather than VDPs for two main reasons. 
Firstly, as mentioned above, better approximation algorithms exist for finding maximum EDP sets
than for finding maximum VDP sets on the grid.
Although, in practice, we make use of the greedy $\bigo{\sqrt N}$-approximation algorithm for finding maximum EDP sets in this work.
Secondly, as was also mentioned above, the maximum EDP set is at least as large as the maximum VDP set.

An important open problem that could ultimately influence the performance of the surface code compilation algorithm we present in this work
is whether an alternative approximation algorithm for finding maximum EDP sets can be used that performs better in practical instances. 

\subsection{Long-range \cnot{} subroutines using VDP and EDP}\label{sec:long-range-cnot-bell}
Here we present one of our main technical contributions, namely a description of how 
to implement a set of long-range \cnot{}s at the end of VDP and EDP sets using surface code operations.
This is central to our overall surface code compilation algorithm presented in \cref{sec:EDPCalogrithm}.

Consider the $L \times L$ square grid graph $G$ (see \cref{fig:graphs-grid}),
which consists of vertices $V(G) = [L] \times [L]$, for $[L] \coloneqq \set{1,\dots, L}$
and undirected edges
\begin{multline}
		E(G) = \set{((i,j),(i,j+1)) \mid i \in [L], j\in [L-1]} \\
		\cup \set{((i,j),(i+1,j)) \mid i\in[L-1], j\in[L]}.
\end{multline}
Here, vertices correspond to qubits stored in surface code patches, and edges connect qubits on adjacent patches (see \cref{fig:logical-layout}).
We color the vertices of $G$ with three colors: black, grey, and white (see \cref{fig:edpSchedule}).
All vertices with both even row and even column index are colored black and correspond to data qubits (where data qubits correspond to qubits in the input circuit).
The vertices (corresponding to ancilla qubits) with both odd row and odd column index are colored white,
and all remaining vertices are colored grey.
This gives us a $1:3$ data qubit to ancilla qubit ratio.
We set $n$ to equal the number of black vertices,
i.e., the number of data qubits.

Due to the designation of some vertices as data qubits and others as ancilla vertices in our layout, and due to the asymmetry of two-qubit operations along horizontal and vertical edges in \cref{fig:logical-layout}, we add some restrictions to the paths we consider.
We define an \emph{operator path} to be a path $P = v_1 v_2 \dots v_k$,
for $k \in \mathbb N$,
such that $v_1$ and $v_k$ correspond to data qubits
and its \emph{interior} $v_2 \dots v_{k-1}$ are all ancilla qubits.
Moreover, $v_1$ to $v_2$ must be a vertical edge,
and $v_{k-1}$ to $v_k$ must be a horizontal edge.
Then an \emph{operator VDP (resp.\ EDP) set} is a set of vertex-disjoint (resp.\ edge-disjoint) operator paths.
In addition, we require that the ends of the paths in the operator EDP set do not overlap.
With the coloring assignments of the grid graph $G$,
it is easy to see that the first and last vertex of an operator path are colored black.
In what follows, we show how we can implement \cnot{}s between the data qubits
at the ends of the paths in an operator VDP (EDP) set in constant depth.

First consider an operator VDP set $\mathcal{P}$.
It is straightforward to see that we can simultaneously apply long-range \cnot{}s along each $P \in \mathcal{P}$ as in \cref{fig:longRangeCnot} in depth 2.
We call this the \emph{vertex-disjoint paths subroutine} (VDP subroutine).

\begin{figure}[tb]
    \centering
    \begin{subfigure}[b]{0.48\textwidth}
        \centering
        \includegraphics[width=0.9\textwidth]{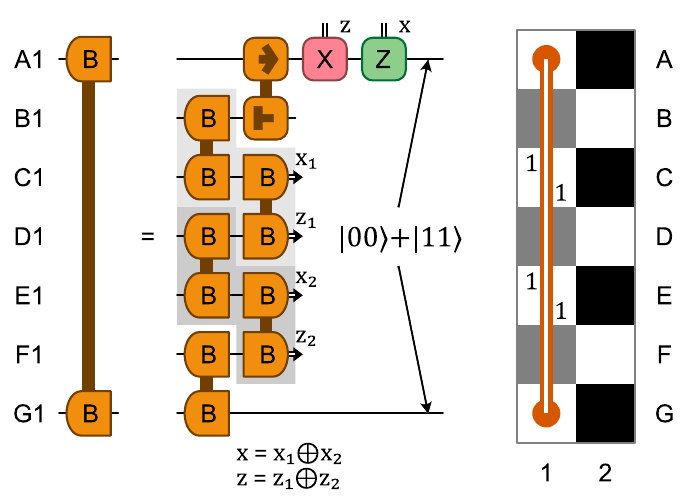}
        \caption{Long-range Bell pair preparation}
        \label{fig:long-range-bell-preparation}
    \end{subfigure}\hfill%
    \begin{subfigure}[b]{0.48\textwidth}
        \centering
        \includegraphics[width=0.8\textwidth]{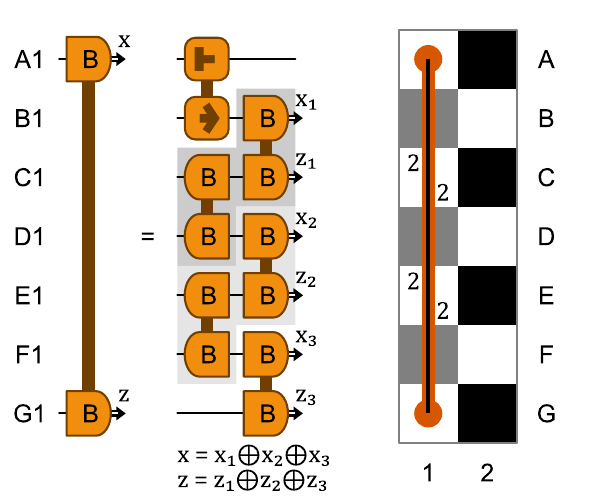}
        \caption{Long-range Bell measurement}
        \label{fig:long-range-bell-measurement}
    \end{subfigure}
    \begin{subfigure}[b]{\textwidth}
        \centering
        \includegraphics[width=0.4\textwidth]{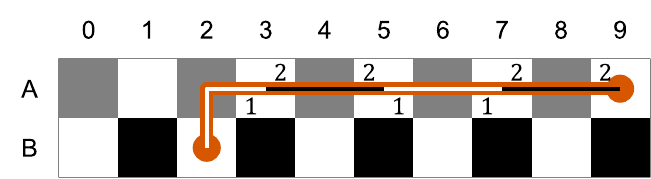}\\
        \includegraphics[width=0.4\textwidth]{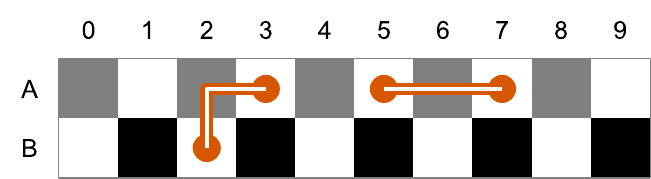}\hfill%
        \includegraphics[width=0.4\textwidth]{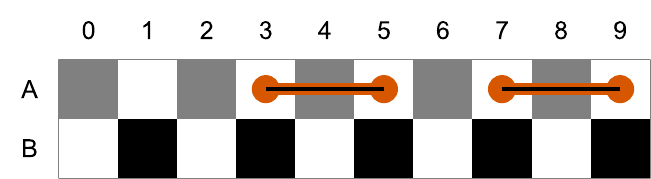}
        \caption{Two-stage Bell preparation using segments}
        \label{fig:long-range-bell-measurement-stages}
    \end{subfigure}    
    \caption{%
    	(\subref{fig:long-range-bell-preparation}) For segments marked in white, we use long-range Bell pair preparation in depth $2$.
    	(\subref{fig:long-range-bell-measurement}) For segments marked in black, we then use long-range Bell pair measurement in depth $2$.
    	(\subref{fig:long-range-bell-measurement-stages}) 	The Bell measurements in stage 2 stitch together the Bell pairs made in phase 1, resulting in a Bell pair in the qubits at the ends of the full path.
	}\label{fig:long-range-bell}
\end{figure}


Now consider an operator EDP set $\mathcal{P}$.
An EDP set can have intersecting paths,
and the ancilla qubits at intersections appear in multiple paths, preventing us from simultaneously producing Bell pairs at their ends.
We circumvent this by producing Bell pairs across a path in two stages by splitting the path into segments; see \cref{fig:long-range-bell}.
We will show that $\mathcal P$ can be \emph{fragmented} into two VDP sets $\mathcal P_1$ and $\mathcal P_2$ that, together, form $\mathcal P$. 
More precisely, each path $P \in \mathcal{P}$ can be built by composing paths contained in $\mathcal{P}_1$ and $\mathcal{P}_2$ such that each path in either $\mathcal{P}_1$ or $\mathcal{P}_2$ appears in precisely one path in $\mathcal{P}$.
We say that the paths in $\mathcal{P}_1$ and $\mathcal{P}_2$ are \emph{segments} of paths in $\mathcal{P}$.
This forms the basis of the \emph{edge-disjoint paths subroutine} (EDP subroutine), which is presented in \cref{alg:edp-protocol}
and illustrated with an example in \cref{fig:edpSchedule}.


\begin{figure}
        \centering
        \begin{subfigure}[b]{0.475\textwidth}
            \centering
            \includegraphics[width=\textwidth]{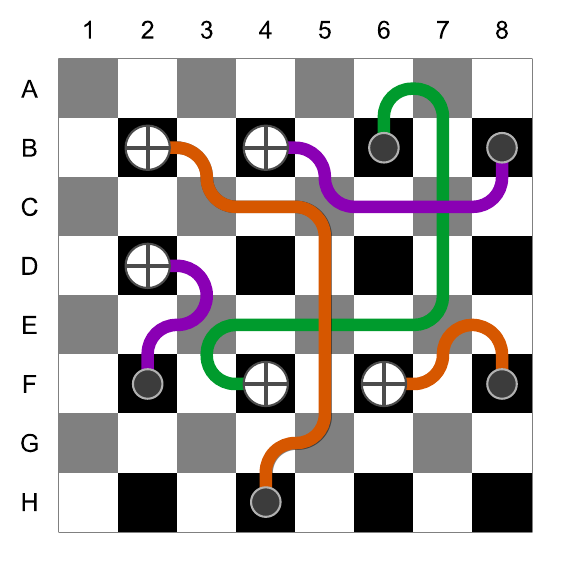}
            \caption{\cnot{} gates and edge-disjoint paths.}
            \label{fig:edge-disjoint-cnots}
        \end{subfigure}%
        \hfill%
        \begin{subfigure}[b]{0.475\textwidth}
            \centering
            \includegraphics[width=\textwidth]{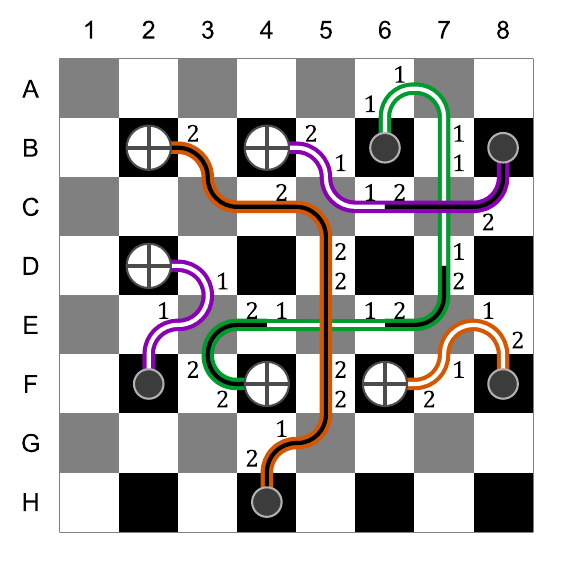}
            \caption{Execution stages assignment}
            \label{fig:edge-disjoint-cnots-assignment}
        \end{subfigure}
        \vskip\baselineskip
        \begin{subfigure}[b]{0.475\textwidth}
            \centering
            \includegraphics[width=\textwidth]{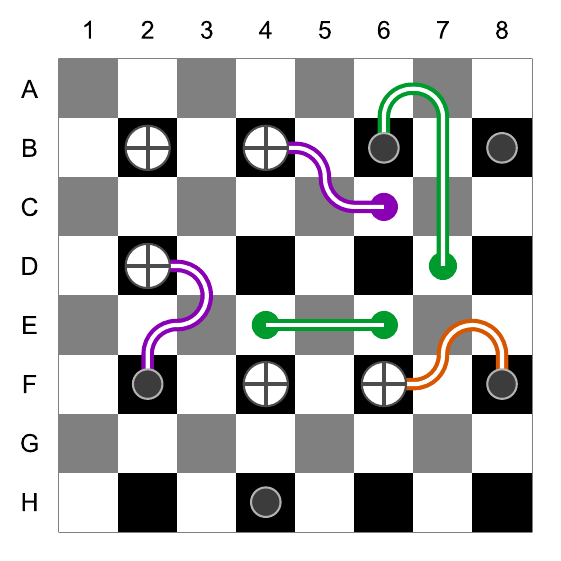}
            \caption{First stage. State preparation.}
            \label{fig:edge-disjoint-cnots-stage-1}
        \end{subfigure}%
        \hfill
        \begin{subfigure}[b]{0.475\textwidth}
            \centering
            \includegraphics[width=\textwidth]{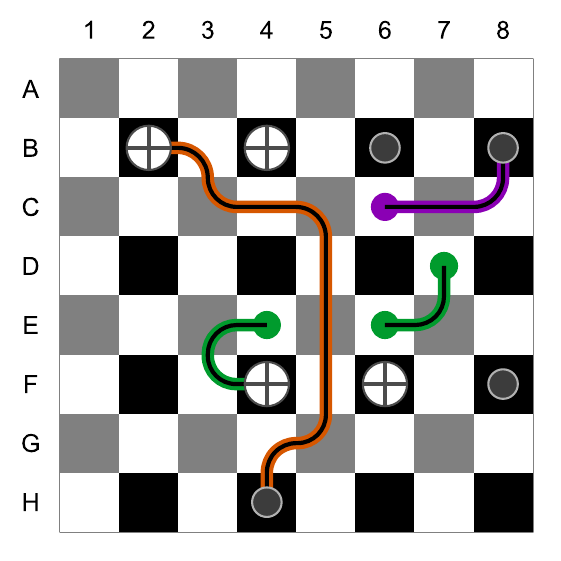}
            \caption{Second stage. Measurements.}
            \label{fig:edge-disjoint-cnots-stage-2}
        \end{subfigure}
    	\caption{%
    		The EDP subroutine implements a set of parallel \cnot{}s connected by an operator EDP set.
    		We assume a qubit ratio of 1 to 3 of data (black) to ancilla (gray and white).
    		(\subref{fig:edge-disjoint-cnots}) The input to the EDP subroutine is a set of \cnot{}s and an associated EDP set.
    		(\subref{fig:edge-disjoint-cnots-assignment}) We fragment the EDP set into two VDP sets consisting of segments of the original paths, and implement the compiled circuit over two depth-2 stages, one for each of these sets.
    		(\subref{fig:edge-disjoint-cnots-stage-1}) During the first stage we prepare a Bell pair between the ends of the segments in the first VDP set.
    		(\subref{fig:edge-disjoint-cnots-stage-2}) During the second stage we perform joint Bell measurements between the ends of segments in the second VDP set,
    		producing long-range Bell pairs on ancillas adjacent to the control and target of each \cnot{}.
        	Then, long-range \cnot{}s can easily be applied by using the long-range Bell pairs (\cref{sec:longRangeCnot}).
        	See \cref{fig:long-range-bell,fig:long-range-cnot} for further details of the long-range operations used here.
    	}\label{fig:edpSchedule}
\end{figure}

\begin{algorithm}[hbt]
	\caption{%
		\emph{EDP subroutine}: to apply \cnot{}s to the data qubits at the endpoints of a set of edge-disjoint paths $\mathcal P$,
		where the interior of each path is supported on ancilla qubits.
		The depth is at most 4.
	}\label{alg:edp-protocol}
	\Input{An operator EDP set $\mathcal P$}
	$\mathcal P_1, \mathcal P_2 \gets$ fragment $\mathcal P$ in two VDP sets of segments\tcp*{\cref{thm:edp-to-phases}}
	\For{segment $P \in \mathcal P_1$}{%
		\eIf{$P$ connects two data qubits}{%
			\execute long-range \cnot{} along $P$\;
		}{%
		\execute phase 1 operation along $P$ (\cref{fig:long-range-bell-preparation}, or \ref{fig:long-range-zz-with-teleport-phase-1}, or \ref{fig:long-range-xx-with-teleport-phase-1})\;
		}
	}
	\For{segment $P \in \mathcal P_2$}{%
		\eIf{$P$ connects two data qubits}{%
			\execute long-range \cnot{} along $P$\;
		}{%
		\execute phase 2 operation along $P$ (\cref{fig:long-range-bell-measurement}, or \ref{fig:long-range-zz-with-teleport-phase-2}, or \ref{fig:long-range-xx-with-teleport-phase-2})\;
		}
	}
\end{algorithm}

\begin{figure}
	\centering
    \begin{subfigure}[b]{\textwidth}
        \centering
        \includegraphics[scale=0.6]{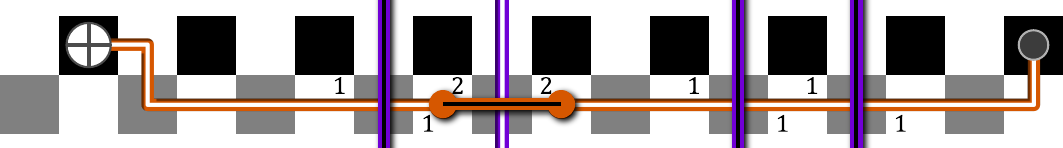}
        \caption{Long-range \cnot{} in two stages}
        \label{fig:long-range-cnot-break-down}
    \end{subfigure}
	\begin{subfigure}[b]{0.475\textwidth}
        \centering
        \includegraphics[width=\textwidth]{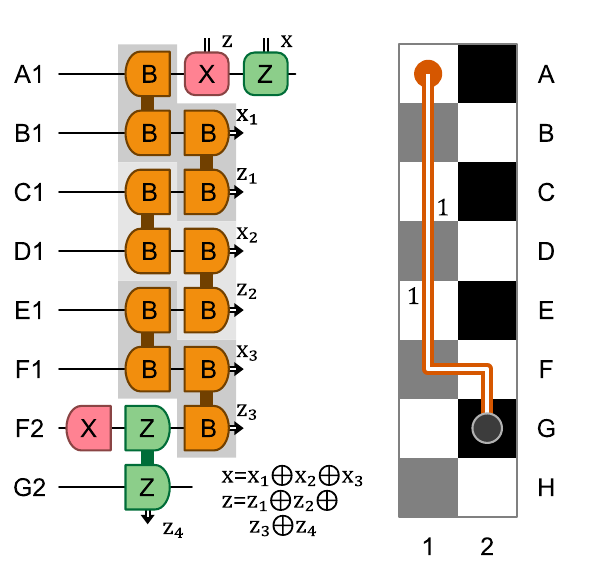}
        \caption{Long-range X prep.\ with ZZ meas.}
        \label{fig:long-range-zz-with-teleport-phase-1}
    \end{subfigure}%
    \hfill%
    \begin{subfigure}[b]{0.475\textwidth}
        \centering
        \includegraphics[width=\textwidth]{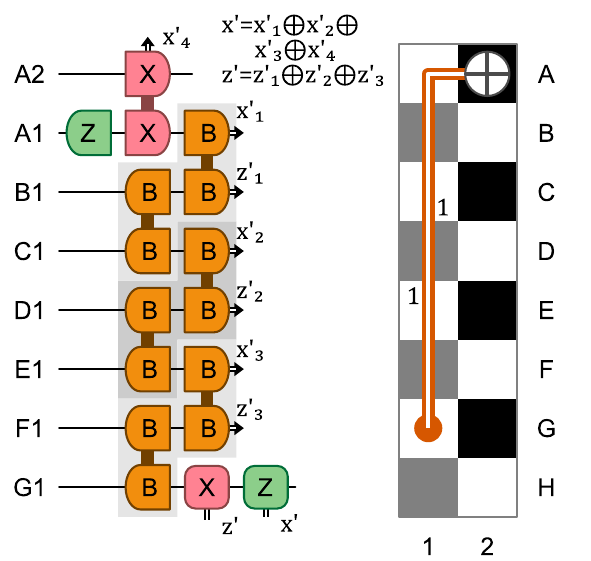}
        \caption{Long-range Z prep.\ with XX meas.}
        \label{fig:long-range-xx-with-teleport-phase-1}
    \end{subfigure}
	\begin{subfigure}[b]{0.475\textwidth}
        \centering
        \includegraphics[width=\textwidth]{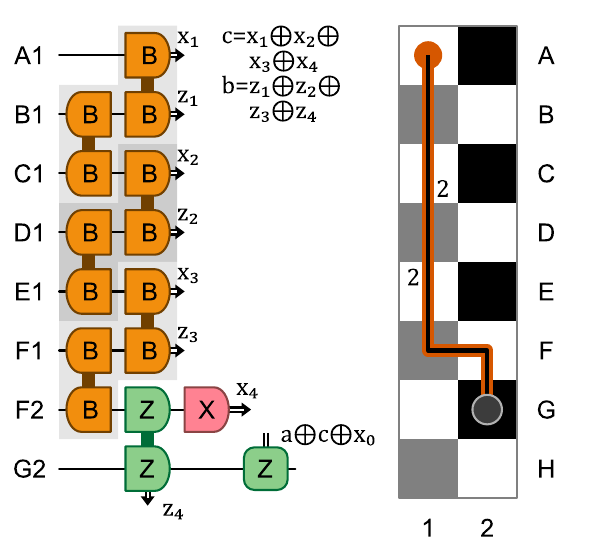}
        \caption{Long-range teleport with ZZ meas.}
        \label{fig:long-range-zz-with-teleport-phase-2}
    \end{subfigure}%
    \hfill%
    \begin{subfigure}[b]{0.475\textwidth}
        \centering
        \includegraphics[width=\textwidth]{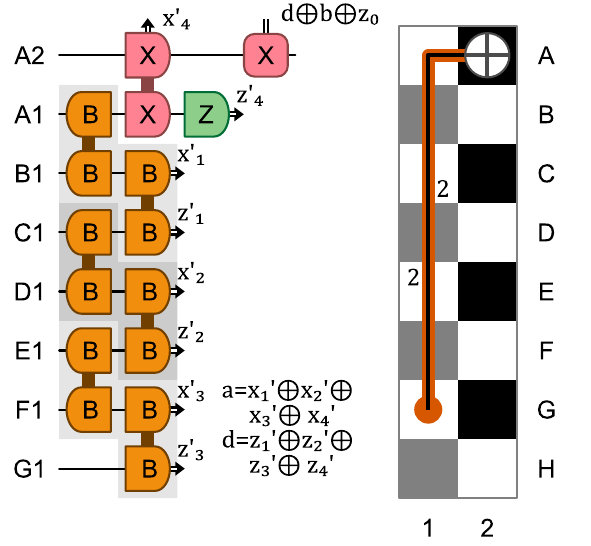}
        \caption{Long-range teleport with XX meas.}
        \label{fig:long-range-xx-with-teleport-phase-2}
    \end{subfigure}
    \caption{%
    	Detailed implementation of the steps in \cref{fig:edpSchedule}.
        For each segment that is scheduled in phase 1,
    	we use (\subref{fig:long-range-zz-with-teleport-phase-1}) and (\subref{fig:long-range-xx-with-teleport-phase-1});
        and for each supbath that is scheduled in phase 2,
    	we use (\subref{fig:long-range-zz-with-teleport-phase-2}) and (\subref{fig:long-range-xx-with-teleport-phase-2}).
        In (\subref{fig:long-range-zz-with-teleport-phase-2}) variables $x_0$ and $z_0$ equal to the total
        parity of all long-range Bell measurements applied during stage 2 on the \cnot{} path.
        Each of these operations takes depth $2$.
        (\subref{fig:long-range-zz-with-teleport-phase-2}) and (\subref{fig:long-range-xx-with-teleport-phase-2}) share the variables $a$ and $c$.
    }\label{fig:long-range-cnot}
\end{figure}

We show the following Lemma,
which restricts the adjacency of \emph{crossing} vertices.
As will become clear later,
the adjacent crossing vertices impose systems of constraints on fragmenting $\mathcal P$,
and their restricted adjacency of any operator EDP set ensures a fragmentation into two VDP sets always exists.

\begin{lemma}\label{lem:crossing-adjacency}
	Given an operator EDP set $\mathcal P$,
	a \emph{crossing} vertex is a vertex contained in more than one path in $\mathcal P$.
	Let the set of crossing vertices be $V_c$,
	then the induced subgraph $G[V_c]$ contains only three kinds of connected components:
	\begin{enumerate}
		\item Isolated vertices.
		\item A \emph{horizontal path},
			where each vertex $(i,j)$ in the connected component can only be adjacent to $(i-1,j)$ and $(i+1,j)$.
		\item A \emph{vertical path},
			where each vertex $(i,j)$ in the connected component can only be adjacent to $(i,j-1)$ and $(i,j+1)$.
	\end{enumerate}
\end{lemma}
\begin{proof}
	We consider all possible colors of a vertex $(i,j)$ in a connected component of $G[V_c]$.
	Black vertices cannot be crossing vertices by definition of an operator EDP set
	so cannot be contained in $V_c$.
	It is then easy to see that white vertices in $V_c$ satisfy the Lemma.

	Therefore, the only relevant case is when $(i,j)$ is a grey vertex.
	The vertices $(i+1,j)$ and $(i,j+1)$ are white and $(i+1,j+1)$ is black.
	We show that these white vertices cannot both be crossing vertices.
	Suppose that they are,
	then both edges between the white vertices and the black vertex,
	$((i+1,j),(i+1,j+1))$ and $((i,j+1),(i+1,j+1))$,
	are in $\mathcal P$.
	This is a contradiction with the fact that
	the interior of operator EDP paths cannot contain a black vertex
	so it must be at the end of two paths,
	but an operator EDP set cannot contain two paths ending at the same vertex.
	By the same argument applied to the other white neighbors of $(i,j)$ we see that only $(i-1,j)$ and $(i+1,j)$ or $(i,j-1)$ and $(i,j+1)$ can both be crossing vertices,
	and the claim follows.
\end{proof}

We now prove that $\mathcal P$ can be fragmented.

\begin{theorem}\label{thm:edp-to-phases}
	We can fragment an operator EDP set $\mathcal P$
	to produce vertex-disjoint sets of segments $\mathcal P_1$ and $\mathcal P_2$.
	If $\mathcal P$ is vertex-disjoint, then $\mathcal P_1 = \mathcal P$ and $\mathcal P_2 = \emptyset$.
\end{theorem}
\begin{proof}
	We assign edges for inclusion in segments in $\mathcal P_1$ or $\mathcal P_2$
	by an edge labelling $l(e) \colon  E(G) \to \set{1,2}$.
	Given a labelling of all edges $e$ in the paths of $\mathcal P$,
	we can assign edges $l(e)=b$ to segments in $\mathcal P_b$.
	Therefore, given a labelling of all edges in paths in $\mathcal P$,
	it is easy to construct $\mathcal P_1$ and $\mathcal P_2$.
	We now label all edge in the paths in $\mathcal P$
	and prove that their labelling
	guarantees the vertex-disjointness property of $\mathcal P_1$ and $\mathcal P_2$.

	We constrain the labeling around every crossing vertex $v$ so that the VDP property is satisfied.
	Clearly, $v$ is contained in the interior of exactly two paths, $P_1$ and $P_2$.
	Let $v$ be contained in edges $e_1$ and $e_1'$ of $P_1$,
	and edges $e_2$ and $e_2'$ of $P_2$,
	then we impose the constraints
	\begin{align}
		l(e_1) &= l(e_1')\\
		l(e_2) &= l(e_2')\\
		l(e_1) &\neq l(e_2)
	\end{align}
	guaranteeing the vertex-disjointness of segments at $v$
	since a segment of $P_1$ must span both $e_1$ and $e_1'$,
	and a segment of $P_2$ must span both $e_2$ and $e_2'$
	with a different label.

	We show there always exists a feasible solution given these constraints.
	If we consider the graph $G[V_c]$ induced by crossing vertices $V_c$,
	then we see that every connected component in $G[V_c]$ gives a system of constraints.
	The adjacency of $G[V_c]$, by \cref{lem:crossing-adjacency},
	is such that each system has one degree of freedom,
	which we decide arbitrarily.

	Finally, for every vertex-disjoint path $P \in \mathcal P$,
	assign $l(e) = 1$ to all edges $e$ in $P$.
	All remaining edges can be labeled arbitrarily.
\end{proof}

The depth of a \cnot{} circuit produced by the EDP subroutine for an operator EDP set $\mathcal P$ is at most 4.
If $\mathcal P$ happens to be vertex-disjoint,
then the depth is 2 since all paths are assigned to phase 1 by \cref{thm:edp-to-phases}.


\subsection{Compiling parallel \cnot{} circuits with the EDP subroutine}\label{sec:EDP-analysis}
In this section we consider how to compile input parallel \cnot{} circuits using the EDP subroutine.
We define the terminal pairs $\mathcal T \subseteq V(G) \times V(G)$ to be the pairs of control and target qubits for each \cnot{} gate in the parallel \cnot{} circuit.
To use the EDP subroutine,
we need to find operator EDP sets $\mathcal P_1, \dots, \mathcal P_k$
that connect all terminal pairs in $\mathcal T$.
We will refer to any such set $\{ \mathcal P_1, \dots, \mathcal P_k \}$ as a \emph{$\mathcal T$-operator set}.
The depth of the compiled implementation is minimized when the size $k$ of the $\mathcal T$-operator set is minimized.

There are reasons to believe that the compilation strategy for parallel \cnot{}-circuits formed by finding a minimal $\mathcal T$-operator set and applying the EDP subroutine should produce low-depth output circuits.
For sparse input circuits, i.e. those with a small number of \cnot{}s,
one can expect a small $\mathcal T$-operator set to exist, giving a low depth output. 
On the other hand, we now prove that there are dense \cnot{} circuits for which the EDP subroutine with a minimal size $\mathcal T$-operator set produces a compiled circuit with optimal depth (up to a constant multiplicative factor).

\begin{theorem}\label{thm:parallel-cnot-not-slow}
	Let a parallel input \cnot{} circuit with corresponding terminal pairs $\mathcal T$ be given,
	and let the $n$ qubits of the input circuit be embedded in a grid among $3n$ ancilla qubits according to the layout in \cref{fig:edpSchedule}.
	For simplicity, we assume $n$ is both even and the square of an integer.
	We can find a $\mathcal T$-operator set of size at most $2\sqrt n - 1$ in polynomial time.
\end{theorem}
\begin{proof}
	For each \cnot{} we construct an operator path
	and argue that all such paths can be grouped into $\bigo{\sqrt n}$ disjoint EDP sets.
	For simplicity, in the following,
	we specify paths by a sequence of key vertices,
	with each consecutive pair of key vertices connected by the shortest path (which is a horizontal or a vertical line).

	We now construct an operator path for each \cnot{},
	where the associated control vertex is $v = (v_x, v_y) \in V(G)$
	and the target vertex is $u = (u_x, u_y) \in V(G)$.
	We can always form an operator path to connect $u$ and $v$ given by the following sequence of five key vertices
	$v$, $(v_x, v_y-1)$, $(u_x-1, v_y-1)$, $(u_x-1, u_y)$, $u$.
	This path consists of one vertical end segment,
	one horizontal interior segment,
	one vertical interior segment,
	and finally a horizontal end segment.

	Having assigned a path to each \cnot{},
	we now show that any of these operator paths can share an edge with at most $2(\sqrt n -1)$ of the other paths.
	Since the operator paths have distinct endpoints,
	two different paths cannot share an edge on either of their end segments
	$v, (v_x, v_y-1)$ and on $(u_x-1, u_y), u$.
	Therefore pairs of these operator paths can only share an edge on their interior segments.
	The horizontal interior segment of the operator path from $v$ to $u$ can share an edge with at most $\sqrt n - 1$ other paths.
	To see this, consider an operator path from $v' = (v'_x, v'_y) \in V(G)$ to $u' = (u'_x, u'_y) \in V(G)$ 
	that shares at least one horizontal edge with the operator path from $v$ to $u$.
	Explicitly, that means the segment $(v_x, v_y-1)$, $(u_x-1, v_y-1)$ shares an edge with
	the segment $(v'_x, v'_y-1)$, $(u'_x-1, v'_y-1)$,
	which implies that $v_y = v'_y$.
	Since the terminals are unique, there can only be $\sqrt n - 1$ other \cnot{}s 
	with the control sharing the $v_y$ coordinate.
	An analogous argument applies for vertical segments,
	such that the operator path from $u$ to $v$ can share an edge with at most $2(\sqrt n - 1)$ other operator paths.

	Let us construct a graph $H$ where each vertex represents an operator path as constructed above.
	We connect two vertices in $H$ if the associated paths share an edge.
	Every vertex in $H$ has degree at most $2(\sqrt n-1)$,
	therefore, $H$ is $(2\sqrt n-1)$-colorable using the (polynomial time) greedy coloring algorithm.
	We construct a $\mathcal T$-operator set of size $2\sqrt n -1$
	by grouping the paths associated with each color in a set of edge-disjoint paths.
\end{proof}

We now show a general lower bound on compiling parallel \cnot{} circuits to the surface code architecture.
Our strategy will be to consider a parallel \cnot{} circuit
with control data qubits in an area with small boundary
that generates an amount of entanglement across the boundary proportional to the area for a given initial state.
However, each elementary surface code operation is local such that only those operations acting at the boundary can increase the entanglement across it.
The depth of any implementation of the \cnot{} circuit is then lower bounded by the entanglement that it generates
over the boundary size~\cite{Delfosse2021,QuantumRouting}.

\begin{theorem}\label{thm:parallel-cnot-lowerbound}
	Consider a surface code architecture of $n$ data qubits
	embedded in a grid where all ancilla qubits are in the $\ket{0}$ state.
	For any positive integer $k \le n/2$,
	there exists a parallel \cnot{} circuit of $k$ \cnot{} gates with associated terminal pairs $\mathcal T$
	that needs depth $\bigomega{\sqrt k}$ to be implemented on the surface code architecture.
\end{theorem}
\begin{proof}
	Consider a \cnot{} circuit with terminal pairs $\mathcal T$
	with control qubits on data vertices in a square region, $V_L$, and target qubits on vertices outside $V_L$.
	We initialize the $2k$ data qubits associated with $\mathcal T$ to a product state $\ket{+}^k \ket{0}^k$,
	with $\ket{+}$ on control qubits
	and $\ket{0}$ on target qubits
	(the remaining data qubits are initialized in an arbitrary product state and ignored).
	After applying the \cnot{} circuit, we obtain $k$ Bell pairs.
	Therefore, the (von Neumann) entropy of the reduced state of the data qubits in $V_L$
	has increased from $0$ to $k$.

	Consider a circuit $\mathcal C$ of depth $d$ that implements the parallel \cnot{} circuit.
	Any elementary operation of the surface code acting only within $V_L$
	or within $\bar V_L \coloneqq V(G) \setminus V_L$
	or classical communication (together, LOCC)
	cannot increase the entropy of the state on $V_L$.
	Moreover, as we show below, each elementary operation that acts both on $V_L$ and on $\bar V_L$ can increase the entropy by at most a constant $4$.
	We can therefore upper bound the increase in entropy due to $\mathcal C$
	by $4 d$ times the number of vertices adjacent to $V_L$, which is proportional to $\sqrt{k}$.
	To attain the $k$ increase in entropy,
	we therefore need that $d = \bigomega{\sqrt{k}}$.
	
	We now bound the increase in entropy of any elementary operations acting
	on $V_L$ and $\bar V_L$ to at most $4$.
	All such elementary operations are built from a single $XX$ or a $ZZ$ measurements
	and single qubit operations~\cref{sec:surfaceCode},
	which cannot increase the entropy.
	It is possible to implement $XX$ and $ZZ$ measurements
	acting on $V_L$ and $\bar V_L$ using two \cnot{}s and 
	operations acting only within $V_L$ or within $\bar V_L$.
	The increase in entropy in $V_L$ by a \cnot{} operation
	is bounded by $2$~\cite[Lemma~1]{Bennett2003}.
	Therefore, $XX$ measurements, $ZZ$ measurements, and indeed 
	any elementary operation of the surface code can increase the entropy by at most $4$.
\end{proof}

In practice, it can be difficult to find minimal-size $\mathcal T$-operator sets.
However, when the minimal size $\mathcal T$-operator set is $k$,
in the following theorem we show that a $\mathcal T$-operator set $\{\mathcal{P}_1,\dots,\mathcal{P}_l\}$
with size at most $l = \bigo{k \log \abs{\mathcal T}}$ can be found by a greedy algorithm that finds
iteratively finds the maximum operator EDP set for remaining terminals in $\mathcal T$.

\begin{theorem}\label{thm:min-operator-sets}
	On the grid of $n$ vertices,
	the greedy algorithm for finding $\mathcal T$-operator sets
	repeats the following two steps , for $i=1,\dots$,
	until there are no more terminal pairs to connect:
	\begin{enumerate}
		\item find a maximum operator EDP set $\mathcal P_i$,
		\item remove all terminal pairs in $\mathcal P_i$ from $\mathcal T$.
	\end{enumerate}
	The set $\{ \mathcal P_1, \dots, \mathcal P_k \}$ is a $\mathcal T$-operator set
	and	is an $\bigo{\log \abs{\mathcal T}}$-approximation
	algorithm for finding minimum-size $\mathcal T$-operator sets.
\end{theorem}
\begin{proof}
	We base our proof on~\cite{Aumann1995}.
	Assume that the minimum-size $\mathcal T$-operator set is $\set{\mathcal Q_1, \dots, \mathcal Q_K}$
	for some size $K$.
	Then there is an operator EDP set $\mathcal Q_i$, for $i \in [K]$,
	such that $\abs{\mathcal Q_i} \ge \abs{\mathcal T}/K$.
	Therefore, the number of unconnected terminal pairs
	is reduced by at least a factor $(1-1/K)$ each iteration
	and it will require at most $\bigo{K \log \abs{\mathcal T}}$ iterations to connect all terminal pairs~\cite{Johnson1974}.
\end{proof}

\begin{figure}
	\begin{subfigure}{0.49\textwidth}
		\centering
		\includegraphics[width=0.8\textwidth]{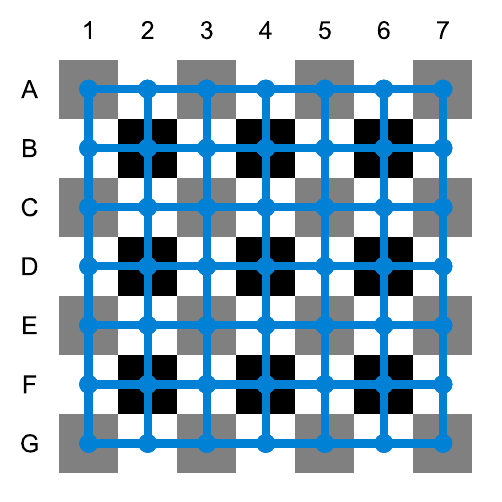}
		\caption{%
			Grid graph
		}\label{fig:graphs-grid}
	\end{subfigure}%
	\hfill%
	\begin{subfigure}{0.49\textwidth}
		\centering
		\includegraphics[width=0.8\textwidth]{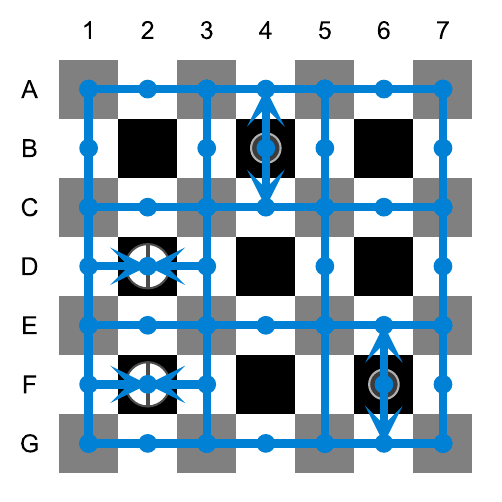}
		\caption{%
			Operator graph
		}\label{fig:graphs-operator}
	\end{subfigure}%
	\caption{%
		The graphs used in this paper.
		(\subref{fig:graphs-grid}) The grid graph where
		each surface code patch corresponds to a vertex and is connected to its neighbors.
		(\subref{fig:graphs-operator}) The operator graph
		for a set of terminal pairs $\mathcal T$ that correspond with a parallel \cnot{} circuit.
		EDP sets for $\mathcal T$ on this graph are also operator EDP sets.
	}\label{fig:graphs}
\end{figure}

To make use of \cref{thm:min-operator-sets} we would ideally like to have an algorithm to find maximum operator EDP sets on the grid,
however the efficient algorithms we discussed in \cref{sec:edp-background} fall short of this in two ways. 
Firstly they find EDP sets rather than operator EDP sets, and secondly they provide approximate maximum sets rather than maximum sets.
Fortunately, we find an equivalence between operator EDP sets on the grid and EDP sets on a graph that we call  \emph{the $\mathcal T$-operator graph} (see \cref{fig:graphs-operator}).
The $\mathcal T$-operator graph is a copy of the grid graph
but with all vertices corresponding to control qubits in $\mathcal T$ only having vertical outgoing edges,
and with all vertices corresponding to target qubits in $\mathcal T$ only having horizontal incoming edges,
and all remaining vertices corresponding to data qubits are removed.
An EDP set for terminal pairs $\mathcal T$ on the $\mathcal T$-operator graph
is an operator EDP set on the grid.
It is easy to see that a maximum operator EDP set for $\mathcal T$ on the grid
is equivalent to a maximum EDP set for $\mathcal T$ on the $\mathcal T$-operator graph.
Using an approximation algorithm for finding the maximum operator EDP set
also still gives approximation guarantees for minimizing the $\mathcal T$-operator set,
as shown in the following Corollary.

\begin{corollary}
	The greedy algorithm for finding minimum $\mathcal T$-operator sets,
	but with a $\kappa$-approximation algorithm for finding maximum operator EDP sets,
	gives an $\bigo{\kappa \log \abs{\mathcal T}}$-approximation algorithm for finding minimum $\mathcal T$-operator sets.
\end{corollary}
\begin{proof}
We modify the proof of \cref{thm:min-operator-sets}
such that every iteration we connect a $(1-\kappa/K)$ fraction of unconnected terminal pairs
using the $\kappa$-approximation algorithm for finding maximum operator EDP sets.
Therefore we obtain a $\bigo{\kappa \log \abs{\mathcal T}}$-approximation algorithm for findining minimum $\mathcal T$-operator sets.
\end{proof}

The equivalence between operator EDP sets on the grid and EDP sets on the $\mathcal T$-operator graph motivates us to seek an efficient algorithm to find approximate maximum EDP sets on the $\mathcal T$-operator graph as a key part of our EDPC algorithm.
The algorithms we discussed in \cref{sec:edp-background} come close to doing this,
but some of them are intended for finding approximate maximum EDP sets on the grid rather than on the $\mathcal T$-operator graph and even if they are adapted,
the guarantees of the size of the approximate minimum EDP sets they produce may not apply in the case of the $\mathcal T$-operator graph.
The algorithms described in Refs.~\cite{Aumann1995,Kleinberg1995} for finding approximate maximum EDP sets on the grid do not directly apply to the operator graph.
While it seems straightforward to adapt the $\bigo{\log n}$-approximation algorithm~\cite{Aumann1995},
the algorithms in~\cite{Aumann1995,Kleinberg1995} are complex to implement and have large constant-factor overheads,
which can make them impractical on small instance sizes.

\begin{algorithm}[tbp]
	\caption{%
		\emph{Bounded $\mathcal T$-operator set algorithm}:
		An approximation algorithm for minimizing the $\mathcal T$-operator set size
		that combines the theoretical guarantees from \cref{thm:parallel-cnot-not-slow}
		with pragmatic performance using the greedy algorithm of \cref{thm:min-operator-sets}.
	}\label{alg:boundedTOperator}
	\Input{%
		$\mathcal T$ terminal pairs
	}
	$\mathcal Q_1 \gets$ the $\mathcal T$-operator set given by \cref{thm:parallel-cnot-not-slow} for $\mathcal T$\label{line:edpc-q1}\;
	$\mathcal Q_2 \gets \emptyset$\;
	\While(\tcp*[f]{Greedy apx.\ minimum-size $\mathcal T$-operator set}){$\mathcal T \ne \emptyset$}{%
		$\mathcal P \gets$ approximately maximize operator EDP set using greedy EDP algorithm~\cite{Kolliopoulos2004} on operator graph\;
		remove connected terminal pairs in $\mathcal P$ from $\mathcal T$\;
		$\mathcal Q_2 \gets \mathcal Q_2 \cup \set{\mathcal P}$\;
	}
	\Return minimum-size set between $\mathcal Q_1$ and $Q_2$\;
\end{algorithm}

In EDPC, we instead combine the theoretical worst-case bounds of \cref{thm:parallel-cnot-not-slow}
with the pragmatic performance of a greedy approach,
which does not have a large constant overhead,
in \cref{alg:boundedTOperator}.
By \cref{thm:parallel-cnot-lowerbound} this gives us asymptotically tight performance in the worst-case.
The runtime of this algorithm is dominated by $\bigo{\abs{\mathcal T}}$ iterations
of approximately maximizing the operator EDP set in time $\bigo{\abs{\mathcal T} n \log n}$.
We leave it as an open question to find better approximation algorithms for finding maximum operator EDP sets
that give improved performance outside the worst-case
and that may also improve the runtime since less iterations over $\mathcal T$ are required.

\section{Remote rotations with magic states}\label{sec:magicState}

Thus far we have discussed the surface code compilation of all the input circuit operations listed in \cref{sec:intro}
except for the single-qubit rotation gates $S=Z(\pi/4)$ and $T = Z(\pi/8)$.
In this section we design a subroutine for the compilation of parallel rotation circuits.
The $S$ and $T$ gates can be implemented by using specially prepared \emph{magic states} $\ket S$ and $\ket T$, respectively.
Magic states can be prepared using a highly-optimized process known as \emph{magic state distillation}~\cite{Knill2004},
which distills many faulty magic states that are easy to prepare into fewer robust states. 
Still, producing both $\ket S$ and $\ket T$ involves considerable overhead.
The $\ket S$ state is used to apply the $S$-gate in a `catalytic' fashion,
whereby the state $\ket S$ is returned afterwards.
On the other hand, the state $\ket T$ is consumed to apply the $T$-gate. 
The reason for this distinction is rooted in the fact that the $S$-gate is Clifford but the $T$-gate is non-Clifford.

In this work, we do not address the mechanism by which magic states are produced,
but instead assume that these states are provided at specific locations where they can be used to implement gates.
More specifically, we assume rotation gates $S$ and $T$ (and also Clifford variations of these such as $X(\pi/8) = T_x$ and $X(\pi/4) = S_x$) can be applied as a resource on specific ancilla qubits $B \subseteq V(G)$
at the boundary of a large array of logical qubits (\cref{fig:delayed-remote-to-the-boundary}).
This will allow sufficient space outside the boundary where highly-optimized magic state distillation and synthesis circuits can be implemented.
Because a large number of magic states are used in the computation,
we consider having magic state distillation adjacent to and concurrent with computation we are concerned with in this paper
to be a reasonable allocation of resources.

\begin{figure}[tb]
	\centering
	\begin{subfigure}[b]{0.475\textwidth}
    	\centering
    	\includegraphics[scale=0.6]{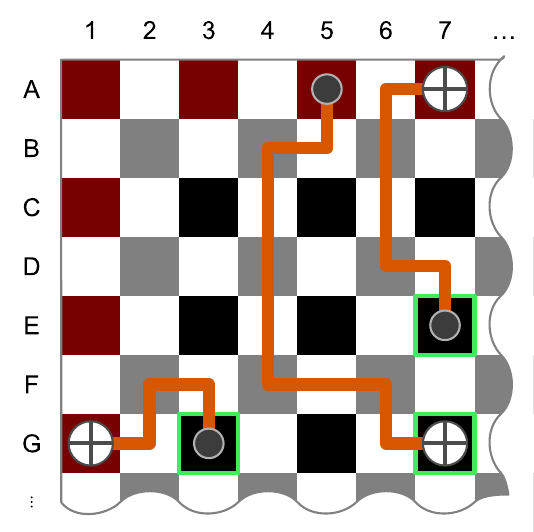}
    	\caption{Long-range \cnot{}s for diagonal gates}%
    	\label{fig:delayed-remote-to-the-boundary}
	\end{subfigure}%
	\hfill%
	\begin{subfigure}[b]{0.475\textwidth}
    	\centering
    	\includegraphics[scale=0.6]{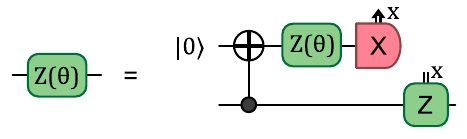}
    	\includegraphics[scale=0.6]{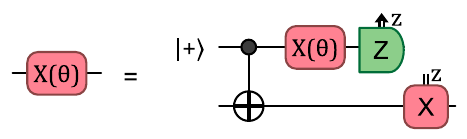}
    	\caption{Remote $Z(\theta)$ and $X(\theta)$}%
    	\label{fig:delayed-remote-gate}
	\end{subfigure}
	\caption{%
			We assume the capability of performing $S$ and $T$ gates at the boundary qubits (red)
        	where it is easy for us to supply the requisite $S$ and $T$ magic states.
			We can then execute $S$ or $T$ gates in the $Z$ or $X$ basis
			for our circuit by using long range \cnot{}s and the circuits in \cref{fig:delayed-remote-gate}.
        	For example, to execute $S$ or $T$ on qubits G3, E7 and $HTH$ on G7, we apply long-range
        	\cnot{}s between pairs (G3,G1), (E7,A7), (A5,G7) and then execute $S$ or $T$ on G1, A7, HTH on A5.
        	We can continue applying other Clifford gates to qubits G3, E7, and G7
        	right after performing the long-range \cnot{}, without waiting for the $Z$ correction,
        	since we can propagate the correction through Clifford operations.
		}\label{fig:delayed-remote}
\end{figure}

We need a technique to apply \emph{remote rotations} to data qubits which can be far from the boundary making use of the rotations that can be performed at the boundary.
We make use of the property that any $Z$ rotation (including $T$ or $S$) has the same action when applied to either qubit in
the state $\alpha\ket{00} + \beta\ket{11}$.
In particular, these two qubits need not be close to one another,
so we can apply $Z$ rotations \emph{remotely}.
A similar notion holds for $X$-rotations (including $T_x = HTH$ or $S_x = HSH$) and $\alpha\ket{++} + \beta\ket{--}$.
Given a qubit $q$ that needs to perform a $Z$ rotation requiring a magic state,
we apply a \emph{remote $Z$-rotation} (\cref{fig:delayed-remote-gate}):
by performing a long-range \cnot{$(q,q')$} to a boundary ancilla $q' \in B$ prepared in $\ket{0}$.
Therefore we can apply the $Z$ rotation remotely and use an $X$ measurement on $q'$ to collapse the state back to one logical qubit.
Similarly, for a qubit $q$ that needs to perform an $X$ rotation requiring a magic state,
we apply a long-range \cnot{$(q',q)$} to an ancilla $q'$ prepared in $\ket{+}$ on the boundary,
giving
\begin{equation}
    \cnot{}\ket{+}(\alpha \ket + + \beta \ket -) = \alpha \ket{++} + \beta \ket{--}\,.
\end{equation}
Therefore we apply the $X$ rotation remotely and collapse the state back by a single-qubit $Z$ measurement of $q'$.

The task of compiling a parallel rotation circuit therefore reduces to applying a set of \cnot{} gates from the boundary to the sites of the rotation gates.
This can be achieved by finding an appropriate EDP set and running the EDP subroutine of \cref{alg:edp-protocol}. 
Compared to the task of finding an EDP set for parallel \cnot{} gates of \cref{sec:edge-disjoint-path-algos}, there is one simplifying condition here:
Any boundary qubit can be used for each \cnot{} when applying remote rotations.
As we explain below, we can find the maximum EDP set for the compilation of remote rotations
by solving the following (unit) \MaxFlow{} problem~\cite{Kleinberg2006:NetworkFlow}.
\begin{definition}[\MaxFlow{}]
    Given a directed graph $G$ and source and sink vertices $s,t \in V$,
    we wish to find a flow for all edges of $G$, $f(e)\colon E(G) \to \mathbb R$,
    that is skew symmetric, $f((u,v)) = -f((v,u))$,
    and, for $v \in V(G) \setminus \set{s,t}$, must respect the constraints
    \begin{align}
    	f(e) &\leq 1\\ 
    	\text{and} \sum_{u : (v,u) \in E(G)} f((v,u)) &= 0
    \end{align}
    such that the outgoing source flow $\abs{f} \coloneqq \sum_{u : (s,u) \in E(G)} f((s,u))$ is maximized.
\end{definition}
To understand why this yields a maximum EDP, we first point out that a solution for which $f$ has binary values provides an EDP set by building paths from those edges $e$ for which $f(e)=1$.
Moreover, this EDP set must be maximum, because a larger EDP set would imply a larger flow than $f$, which is the maximum flow by definition.
Indeed the Ford-Fulkerson algorithm~\cite{Ford1956} solves \MaxFlow{} in runtime bounded by $\bigo{\abs{E(G)}\abs{f}}$
and finds flow values $f(e) \in \set{0,1}$ on all $e \in E(G)$
because of the unit capacity constraints, $f(e) \leq 1$.
Therefore, $f$ corresponds to a maximum EDP set~\cite[Section~7.6]{Kleinberg2006:NetworkFlow}.

\begin{algorithm}[tbp]
	\caption{%
		\emph{Remote rotation subroutine}: executes parallel single qubit rotations that require magic states at the boundary
		by a \MaxFlow{} reduction.
		Using the EDP subroutine (\cref{alg:edp-protocol}),
		we can perform remote rotations (\cref{fig:delayed-remote})
		on each set of qubits connected to the boundary by $\mathcal P$
		in depth 4.
		}\label{alg:magic-rotations}
	\Input{%
		Connectivity graph $G$ with vertices corresponding to boundary qubits $B \subseteq V(G)$
		and a set of parallel rotations $\mathcal G_m$\;
	}
	\SetKwFunction{KwMaxRotations}{max\_rotations}
	\Fn{\KwMaxRotations{$\mathcal G_m$}}{%
		$W \gets$ vertices associated with qubits in $\mathcal G_m$\;
		create virtual vertices $s$ and $t$\;
		$G' = (V(G) \cup \set{s,t}, E(G) \cup \set{(s, s') \mid s' \in W} \cup \set{(t', t) \mid t' \in B})$\;
		$f \gets$ solve \MaxFlow{} on $G'$ using the Ford-Fulkerson algorithm\;
		$\mathcal P \gets$ construct edge-disjoint set of $s$--$t$ paths from $f$\;
		\Return $\mathcal P$ with $s$ and $t$ removed from each $P \in \mathcal P$\;
	}
	\While{$\mathcal G_m$ is not empty}{%
		$\mathcal P \gets$ \KwMaxRotations{$\mathcal G_m$}\;
		\execute remote rotations at boundary with EDP subroutine given $\mathcal P$\;
		remove executed rotations from $\mathcal G_m$\;
	}
\end{algorithm}

The \emph{remote rotation subroutine} (\cref{alg:magic-rotations}) executes a set of parallel single-qubit rotations.
Each iteration can be performed in depth 4 using the EDP subroutine.
On the surface code architecture,
we can give strong guarantees on the number of iterations required to execute a set of parallel rotations
by the \MaxFlow{} to min-cut equivalence.

\begin{theorem}\label{thm:rotation-depth}
	The remote rotation subroutine executes all rotations in $\mathcal G_m$ in depth $\bigo{\sqrt{\abs{\mathcal G_m}}}$.
\end{theorem}
\begin{proof}
	The function \KwMaxRotations{$\mathcal G_m$} that is a part of the remote rotation subroutine
	finds a maximum flow connecting the data qubits performing rotations to the boundary
	where every additional unit of flow is one more rotation executed.
	This maximum flow is equal to the minimum edge-cut separating the data qubits from the boundary~\cite{Ford1956}.
	The boundary of a rectangle containing $\abs{\mathcal G_m}$ vertices on the grid is of size $\bigomega{\sqrt{\abs{\mathcal G_m}}}$,
	giving a minimum cut size of $\bigomega{\sqrt{\abs{\mathcal G_m}}}$.
	Thus, at most $\bigo{\sqrt{\abs{\mathcal G_m}}}$ iterations
	of the while loop in the remote rotation subroutine are necessary to implement all remote rotations,
	as claimed.
\end{proof}

We bound the runtime of the remote remote rotation subroutine by $\bigo{n^2 \sqrt{\abs{\mathcal G_m}}}$ as follows:
At most $\bigo{\sqrt{\abs{\mathcal G_m}}}$ iteration of the while loop are necessary (see proof of \cref{thm:rotation-depth}).
Each iteration, the call to \KwMaxRotations{$\mathcal G_m$} is dominated by
solving a \MaxFlow{} instance using the Ford-Fulkerson algorithm~\cite{Ford1956},
which has a runtime bounded by $\bigo{n^2}$.

One could consider a number of generalizations and variations of this compilation subroutine for parallel rotation circuits.
For instance, when the number of rotation gates is small,
it may be useful to find VDP sets rather than EDP sets
so that the VDP subroutine rather than the EDP subroutine can be applied.
There is a different reduction to \MaxFlow{} in this case which can be obtained by replacing each vertex with two vertices,
one with incoming edge and one with outgoing edges,
connected by a directed edge with capacity 1.
This guarantees only one flow can pass through every vertex.

Although we do not consider other single-qubit rotations in our input circuit for compilation, it is worth noting that any single-qubit rotation gate $Z(\theta)$ can be approximately synthesized to arbitrary precision~\cite{RossSelinger2014} using $\ket{S}$ and $\ket{T}$ states along with the surface code operations shown in \cref{fig:logical-layout}.
The approach used to apply $S$ and $T$ gates shown in \cref{fig:delayed-remote-to-the-boundary} can also be used to apply any rotation $Z(\theta)$ within the grid of surface codes by synthesizing the rotation at the boundary. 
However, if one considers more general rotations in the input circuit, the time needed for synthesis at the boundary will need to be accounted for and accommodated by other aspects of the overall surface code compilation algorithm.
Another extension that can be considered is if multi-qubit diagonal gates are allowed in the input circuit. 
We show how $X$ and $Z$ rotations generalize to multi-qubit diagonal gates in \cref{sec:RemoteGate},
although we do not use this in our surface code compilation algorithm.

\section{EDPC surface code compilation algorithm}\label{sec:EDPCalogrithm}
In this section we construct the EDPC algorithm for compiling universal input circuits into surface code operations
by combining subroutines \cref{alg:edp-protocol} and \cref{alg:magic-rotations} for compiling long-range \cnot{}s
and $Z/X$ rotations respectively.
First we provide a more formal definition of surface code compilation:
	
\begin{definition}[Surface code compilation]
	Consider an input quantum circuit of operations $\mathcal C = g_1 g_2 \dots g_\ell$,
	which is a list of length $\ell$ of operations $g_i$ for $i \in [\ell]$,
	consisting of:
	state preparation in $X$ or $Z$ basis;
	the single-qubit operators $X$, $Y$, $Z$, $H$, $S$, $T$, $S_x = HSH$, $T_x = HTH$;
	\cnot{} operations;
	and $X,Z$-measurements.
	Then a surface code compilation produces an equivalent output circuit $\mathcal O$
	in terms of surface code operations (\cref{fig:logical-layout}) on a grid of surface codes
	with $S$, $T$, $S_x$, and $T_x$ rotations applied only at the grid's boundary.
\end{definition}

\begin{algorithm}[tbp]
	\caption{%
		\emph{EDPC}: a surface code compilation algorithm for any circuit $\mathcal C = g_1 \dots g_\ell$.
		An operation $g_i$ is \emph{available} if it has not been executed
		and all operations $g_j$ with overlapping support, for $j < i$, are executed.
	}\label{alg:edpc}
	\Input{Circuit $\mathcal C$ with Paulis commuted to the end and merged with measurement}

	\While{available operations in $\mathcal C$}{%
		\execute all available state preparation, measurement, and Hadamard\;

		run remote rotation subroutine on available rotations\;
		$\mathcal T \gets$ terminal pairs associated with available \cnot{}s\;
		$\mathcal Q \gets$ run bounded $\mathcal T$-operator set \cref{alg:boundedTOperator} on $\mathcal T$\;
		\For{$\mathcal P \in \mathcal Q$}{%
			run EDP subroutine (\cref{alg:edp-protocol}) on $\mathcal P$\;
		}
	}
\end{algorithm}

The surface code compilation algorithm EDPC (\cref{alg:edpc})
combines combines the bounded $\mathcal T$-operator set algorithm for parallel \cnot{}s
with the remote rotation subroutine.
Note that the input circuit is considered to be a sequence of operations rather than a series of time steps 
that specify the operations in each time step, such that $l$ is the number of operations of the input circuit, not the depth.

We bound the classical runtime of EDPC
given an input circuit with depth $D$ acting on $n$ qubits.
It is useful to note that each of the $D$ layers of the input circuit can be decomposed into a set of parallel rotations followed by a set of parallel \cnot{}s, each acting on at most $n$ qubits.
Recall that the remote rotation subroutine has a runtime bounded by $\bigo{n^{2.5}}$,
whereas compiling a set of parallel \cnot{}s has a runtime of at most $\bigo{n^3 \log n}$.
Thus, EDPC has a runtime bounded by $\bigo{D n^3 \log n}$.

Circuits compiled by EDPC can be bounded in depth as listed in \cref{tab:comparison}.
Our claim for a single \cnot{} is trivial.
\cref{thm:parallel-cnot-not-slow,thm:parallel-cnot-lowerbound} show that parallel \cnot{} circuits are compiled to a depth of $\bigtheta{\sqrt n}$,
and \cref{thm:rotation-depth} shows that $k$ parallel rotations are compiled to a depth of $\bigo{\sqrt k}$.
It is then easy to see that a circuit of depth $D$ compiles to a circuit of depth at most $\bigo{D\sqrt n}$.
If we assume a remote rotation must be performed for each rotation requiring magic states at the boundary
(in particular, it requires a long-range \cnot{} as in EDPC),
then \cref{thm:parallel-cnot-lowerbound} shows an $\bigomega{\sqrt k}$ lower bound on the depth to apply $k$ \cnot{} operations with the boundary.

There are various modifications of EDPC that are worth considering.
Firstly, the bounded $\mathcal T$-operator set algorithm (\cref{alg:boundedTOperator})
can be improved by better algorithms for finding maximum operator EDP sets.
Secondly, the requirement to execute all available gates before moving on to the next set could be relaxed.
This could increase the number of long-range gates that are performed in parallel
but would require careful scheduling with Hadamard gate execution,
which may block some paths.
Lastly, EDPC leans heavily on finding operator EDP paths and the EDP subroutine,
but a similar surface code compilation algorithm could be constructed from operator VDP paths and the VDP subroutine instead.
We believe that larger maximum EDP sets
allows EDPC to apply more gates simultaneously (see \cref{sec:edp-background}),
and more so if algorithms for approximation maximum operator EDP sets can adopted from EDP approximation algorithms~\cite{Aumann1995,Kleinberg1995}.
Both of these features can give asymptotic improvements at only a $2\times$ depth increase over the VDP subroutine.
However, it is not difficult to construct instances where a VDP-based approach would give a lower depth,
motivating a more nuanced trade-off between our EDP-based approach and a VDP-based approach.

\section{Comparison of EDPC with existing approaches}\label{sec:results}
In this section, we compare EDPC with other approaches in the literature.
We first mention some of the features and short-comings of the well-established approach of Pauli-based computation \cref{sec:pauli-based}.
Then we address a more recently proposed compilation approach based on network coding in \cref{sec:networkCoding}.
In \cref{sec:compilation-swap} we specify a \sw{}-based compilation algorithm~\cite{Childs2019}
and use this as a benchmark for numerical studies of the performance of an implementation of EDPC in \cref{sec:numerical-results}.

\subsection{Surface code compilation by Pauli-based computation}\label{sec:pauli-based}

One well-established surface code compilation approach is known as \emph{Pauli-based computation}, which is described 
in \cite{Litinski2019}. 
For an algorithm expressed in terms of Clifford and $T$ gates,
Pauli-based computation first involves re-expressing the algorithm as
a sequence of joint multi-qubit Pauli measurements along with additional ancilla qubits prepared in $T$ states.
This re-expressed circuit has no Clifford operations,
and the circuit depth can be straight-forwardly deduced from the input circuit since each $T$ gate results in two~\cite{Chamberland2021} joint Pauli measurements\footnote{%
	In the scheme presented in~\cite{Litinski2019} only one joint Pauli measurement is needed per $T$ gate,
	but additional features are required of the surface code such as twist defects
	which were avoided in~\cite{Chamberland2021}, and which we have avoided in this paper.}. 
This re-expression of the circuit essentially comes from first replacing each $T$ gate
by a small gate teleportation circuit consisting of an ancilla in a $T$ state and a two-qubit joint Pauli measurement,
and then commuting all Clifford operations to the end of the circuit.
The main advantage of the Pauli-based computation approach is that all Cliffords are removed from the input circuit,
resulting in no cost for \cnot{} circuits in \cref{tab:comparison}.

That said, this approach has a major drawback.
When a Clifford circuit is commuted through a two-qubit joint Pauli measurement,
it is transformed into Pauli measurements which can have support on all logical qubits.
Therefore, the resulting circuit may contain measurements with large overlapping support
that need to be performed sequentially
(even when the $T$ gates in the input circuit were acting on disjoint qubits during the same time step).
The sequential nature of the joint measurements causes a fixed rate of $T$-state consumption
that does not grow with the number of logical qubits
and results in a $\bigtheta{k}$ depth for $k$ parallel rotations, as listed in \cref{tab:comparison}.
The depth for parallel rotations is signifcantly higher than EDPC
and could lead to a larger space-time cost for circuits with many $T$ gates per time step.


A modified version of this Pauli-based computation compilation algorithm can be used to implement more $T$ gates in parallel~\cite[Section~5.1]{Litinski2019}.
However, as highlighted in~\cite[Section~V.A]{Chamberland2021},
this results in a significant increase of total logical space-time cost
when compared to the standard Pauli-based computation compilation algorithm,
even when disregarding the increased $T$-factory costs that would be needed to achieve a higher $T$ state production rate.

In contrast with Pauli-based computation,
one of our goals when designing the EDPC algorithm was to maintain the parallelism present in the input circuit,
such that input circuits with higher numbers of $T$ gates per time step are compiled to circuits with a higher $T$-state consumption rate.

\subsection{Surface code compilation by network coding}\label{sec:networkCoding}
Another approach to surface code compilation, based on the field known as \emph{linear network coding}~\cite{Ahlswede2000}, can be built from the framework put forward in~\textcite{Beaudrap2020a}.
Similar to our EDPC algorithm, the essential idea in this compilation scheme is to generate sets of Bell pairs in order to implement operations acting on pairs of distant qubits.

In the abstract setting of network coding~\cite{Lehman2004}, one is given a directed graph $G_{\text{NC}}$ and a set of terminal pairs $\mathcal T = \set{(s_1, t_1), \dots, (s_k, t_k)}$ for source terminals $s_i \in V(G_{\text{NC}})$ and target terminals $t_i \in V(G_{\text{NC}})$ for $i \in [k]$.
Messages are passed through edges according to a linear rule. 
Namely, the value of the message associated with an edge is given as a specific linear combination of the values of those edges which are directed at the edge's head. 
One can consider the task of ``designing a linear network code''
by specifying the linear function at each edge in the graph
such that when any messages are input via the source vertices $s_1, \dots, s_k$,
then those same messages are copied over to the corresponding output via the target vertices $t_1, \dots, t_k$.

A number of works have considered how linear network coding theory can be applied to the quantum setting~\cite{Leung2010,Kobayashi2009,Kobayashi2011,Satoh2012,Hahn2019}.
\textcite{Beaudrap2020a} gives a construction for a constant-depth circuit to generate Bell pairs across the terminal pairs $\mathcal T$
on a set of ancilla qubits corresponding to the vertices of $G_{\text{NC}}$ with \cnot{}s allowed on the edges of $G_{\text{NC}}$. 
This is similar to, but not precisely the same scenario as we consider for surface code compilation in this paper since the basic operations are \cnot{}s rather than the elementary operations of the surface code, and since only ancilla qubits are considered without any data qubits.
However, it should be quite straightforward to modify the approach in \textcite{Beaudrap2020a} to form a surface code compilation algorithm.
For example, one could use a layout similar to that which we use for EDPC in \cref{fig:edpSchedule}, with $G_{\text{NC}}$ corresponding to a connected subset of ancilla qubits among a set of data qubits.
The Bell pairs produced by the linear network coding approach could then be used to compile long-range operations between data qubits.

In such a network coding based compilation algorithm, the task of compiling an input circuit into surface code operations would largely rely on subroutines for
(1) identifying $\mathcal T$ to implement the circuit's long-range gates,
and (2) designing a linear network code for $\mathcal T$.
A major barrier to forming a usable compilation algorithm with linear network coding is that we are unaware of the existence of any efficient algorithm to design linear network codes,
or even to identify if a given terminal pair set admits any linear network code.
Even if such a linear network code can be found efficiently,
there exist sets $\mathcal T$ for which network coding cannot provide a depth advantage over EDPC.

Any surface code compilation algorithm of \cnot{} circuits with $k$ parallel \cnot{}s,
including EDPC and algorithms using network coding,
is lower bounded in the worst case by \cref{thm:parallel-cnot-lowerbound} to a depth of $\bigomega{\sqrt k}$.
This bound is loose when $k$ is superconstant and sublinear in $n$
since EDPC has a trivial upper bound of $\bigo{k}$
and a bound of $\bigo{\sqrt n}$ by \cref{thm:parallel-cnot-not-slow}
on the compiled circuit depth.
Therefore, it remains an open question whether network coding
can give an advantage for the compiled circuit depth for such $k$.

\subsection{Surface code compilation by \sw{}}\label{sec:compilation-swap}
Here we specify a \sw{}-based compilation algorithm, stated in \cref{alg:swap-compilation}, which we use to benchmark our EDPC against in \cref{sec:numerical-results}.
We assume the 1-to-1 ancilla-to-data qubit ratio as illustrated in \cref{fig:swapSurfaceCode}. 
This is more qubit-efficient than the 3-to-1 ratio we use for EDPC, and it allows the swap gadget in \cref{fig:swap-via-move} to be implemented between diagonally neighboring data qubits.

\begin{figure}
	\centering
	\includegraphics[scale=0.6]{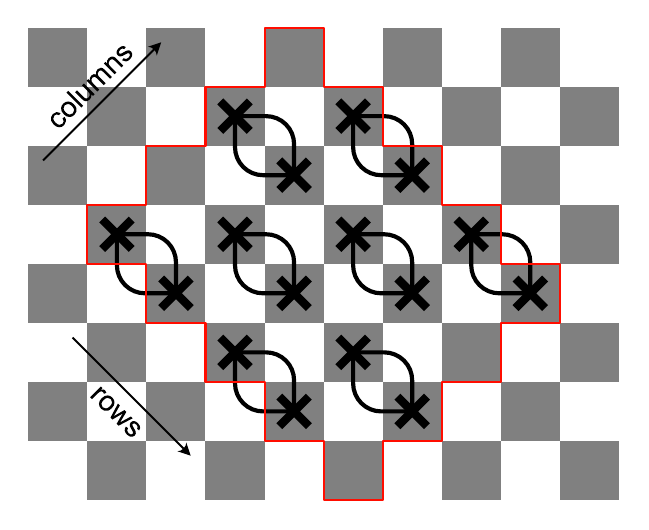}
    \caption[Odd-even pattern of swaps.]{%
    	On a rotated $L_1 \times L_2$ grid (here, $4\times 5$),
    	we can implement an odd-even pattern of swaps on data qubits (gray) using ancillas (white).
        Row-wise and column-wise \sw{}s used in \sw{} routing on a grid~\cite{Alon1994}
        can be modified as shown above so that ancilla used for \sw{}s do not overlap.
        Therefore, any arbitrary permutation on a rotated grid of
        can be implemented in space-time $4(L_1+1) + 2(L_2+1)$.
    }\label{fig:swapSurfaceCode}
\end{figure}

The first step of the \sw{}-based compilation algorithm is to assign each of the input circuit's qubits to a data qubit in the layout.
Then, the gates in the input circuit are collected together into sets of disjoint gates. 
Before each set of gates, a permutation built from \sw{}-gates is applied, which re-positions the qubits so that the gates in the set can be applied locally. 
We assume that the available local operations are the same as for our EDPC algorithm. 
In particular, we assume that the rotation gates ($S$, $T$, $S_x$ and $T_x$) can only be implemented at the boundary and that other single-qubit operations are performed as described in \cref{sec:natural-gates}. 
One exception is that we make the simplifying assumption that the Hadamard can be performed without the need of three ancilla patches to simplify our analysis -- this assumption could lead to an underestimate of the resources required for this \sw{}-based compilation algorithm.

\SetKwFunction{route}{route}
\SetKwFunction{circDepth}{depth}
\SetKwFunction{swCost}{cost}
\begin{algorithm}
	\caption{%
		\emph{\sw{} compilation}: 
		We construct an algorithm based on the \emph{greedy depth} mapper algorithm from~\cite{Childs2019}.
		Let us implicitly define \protect\route{$\pi$}, for mapping $\pi$,
		which finds a \sw{} circuit for implementing partial permutations~\cite{Childs2019}.
		We can compute the required partial permutation from the current mapping of qubits,
		and the given future mapping $\pi$.
	}\label{alg:swap-compilation}
	\Input{A circuit $\mathcal C$ with all Paulis commuted to the end and merged with measurement}
	\Fn{\swCost{mapping $\pi$, vertices $v_1, v_2$}}{%
		\Return depth and edge attaining $\min_{e \in \mathcal M} \text{\circDepth{\route{$\pi + \set{v_1 \mapsto e_1, v_2 \mapsto e_2}$}}}$
	}

	\While{available gates in $\mathcal C$}{%
		$\mathcal G \gets$ available gates in $\mathcal C$\;
		\execute all Hadamards and measurements in $\mathcal G$\;

		$G \gets$ surface code grid graph\;
		$\pi \gets$ empty mapping of $V(G) \to V(G)$\;
		\tcp{Start modification for operations requiring magic states}
		Set $B \subseteq V(G)$ as the set of boundary vertices\;
		$\mathcal G_m \gets \set{g \in \mathcal G \mid g \text{ is } S, T, S_x, T_x}$\;
		\While{$\mathcal G_m \not= \emptyset$ and $B \not= \emptyset$}{%
			$g \gets$ pop random gate from $\mathcal G_m$\;
			$\pi \gets \pi + \set{v \mapsto u}$, for closest $u \in B$ to $v$\;
			remove $u$ from $B$ and $G$\;
		}

		\tcp{End modification}
		$\mathcal M \gets$ maximum matching of $G$\;
		$\mathcal G_c = \set{g \in \mathcal G \mid g \text{ is } \cnot{}}$\;
		\While{$G_c \not= \emptyset$ and $\mathcal M \not= \emptyset$}{%
			$g^*, e^* \gets \max_{g \in \mathcal G_c}$ \swCost{$\pi + \set{v_1 \mapsto e_1, v_2 \mapsto e_2}$} for $v_1,v_2$ current location of $g$\;
			$\pi \gets \pi + \set{v_1 \mapsto e_1, v_2 \mapsto e_2}$, for $v_1,v_2$ current location of $g^*$\;
			remove $g^*$ from $\mathcal G_c$\;
			remove $e^*$ from $\mathcal M$\;
		}

		\execute the \sw{}s found by \route{$\pi$}\;
		\execute gates on qubits mapped by $\pi$ since they are now local\;
	}

\end{algorithm}

There are two main components of our \sw{}-based algorithm which remain to be specified: how the permutations are implemented, and how we choose to separate the input circuit into a sequence of sets of disjoint gates.
To permute the positions of data qubits, sequences of \sw{} operations are used. 
Any permutation of the $n$ vertices in a square grid can be achieved in at most $3\sqrt{n}$ rounds of nearest-neighbor swaps~\cite{Steiger2018}.
To do this involves three stages, with the first and third stages each involving rounds of \sw{}-gates within rows only, and the second stage involving rounds of \sw{}-gates within columns only. 
A round of \sw{}-gates within either rows only or within columns only are implemented with surface code operations as shown in \cref{fig:swapSurfaceCode}.
This immediately shows that this approach is asymptotically tight for parallel circuits
because the depth of a \sw{}-based approach is lower bounded by the $\sqrt{n}$ diameter of the architecture grid
for one long-range \cnot{} or rotation gate from the center of the grid.
Therefore a parallel input circuit is compiled by the \sw{}-based algorithm to an output circuit with depth $\bigtheta{\sqrt{n}}$, including all the examples in \cref{tab:comparison}.

There is considerable freedom in how to collect together gates from the input circuit into sets of disjoint gates. 
In our implementation in \cref{alg:swap-compilation}, we use the \emph{greedy depth} mapper algorithm from~\cite{Childs2019}, with a small modification to ensure that $S$ and $T$ gates are performed at the boundary. 
This algorithm also incorporates some further optimizations as described in~\cite{Childs2019}, including a partial mapping of qubits to locations,
leaving the remaining qubits to go anywhere in an attempt to minimize the \sw{} circuit depth.

\subsection{Numerical results}\label{sec:numerical-results}
\begin{figure}[tb]
	\includegraphics[width=\linewidth]{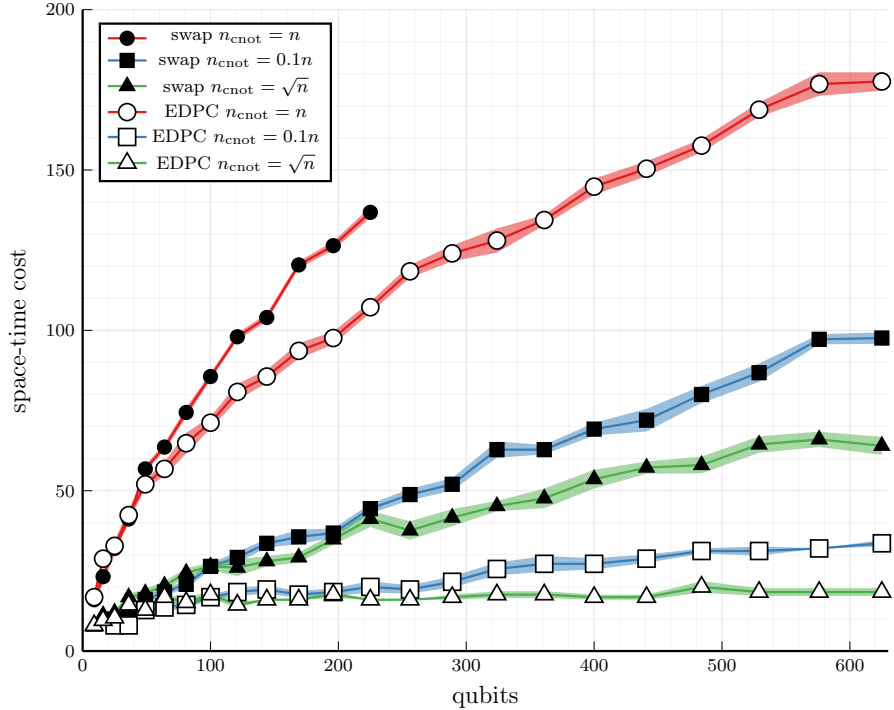}
	\caption{%
		Space-time cost of a randomly sampled set of disjoint \cnot{}s with standard error of the mean (shaded region)
		compiled to the surface code using EDPC and \sw{} compilation.
		We generate 10 random circuits for each number of qubits ($n$) consisting of a set of disjoint \cnot{}s of varying density;
		the number of randomly selected qubits involved in a \cnot{} is given by $n_{\cnot{}}$.
		At all densities we see improved performance and scaling using EDPC.
	}\label{fig:benchmark-random-circuits}
\end{figure}

Here we numerically compare the performance of EDPC with the \sw{}-based compilation algorithm (\cref{alg:swap-compilation}) when applied to a number of different input circuits.
Note that our implementation of the EDPC compilation algorithm here differs slightly from that given in \cref{alg:edpc},
by greedily executing \cnot{}s earlier where possible.
See \cref{app:edpcImplementation} for details of the implementation.

Our first input circuit example consists of random parallel \cnot{} circuits of different gate densities.
The density $n_\cnot{}$ of a circuit is how many of the data qubits are involved in a \cnot{} gate in any such set.
Therefore, $n_\cnot{} = 0.1 n$ means that $10\%$ of all qubits ($n$) are performing a \cnot{} gate in each set.
For each data point, we sample 10 random circuits and plot the mean space-time cost in \cref{fig:benchmark-random-circuits}
with the standard error of the mean in the shaded region.
The runtime of the \sw{} protocol was bounded by 2 days, which was insufficient for larger instances of these random circuits at high densities.

\begin{figure}[tbp]
	\begin{subfigure}[b]{0.3\textwidth}%
	\centering%
	\includegraphics[width=\textwidth]{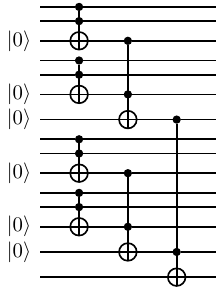}
	\caption{Half \cknot{8}}\label{fig:multi-controlled-x-half-circuit}
	\end{subfigure}%
	\begin{subfigure}[b]{0.7\textwidth}
		\centering
		\includegraphics[width=\textwidth]{MultiControlledXHalfEstimator.tikz}
		\caption{%
			Compilation of half \cknot{k}
		}\label{fig:multi-controlled-x-half-results}
	\end{subfigure}
	\caption{%
		We compare the space-time cost of compiling a $T$-gate optimized circuit decomposition
		for a half \cknot{k} circuit to the surface code
		using EDPC and \sw{} compilation.
		We see in the log-log plot (\subref{fig:multi-controlled-x-half-results})
		that dependence of the space-time cost on $n$
		gives a higher scaling dependence in the case of \sw{} compilation than EDPC.
		This results a lower space-time cost for EDPC starting from 64 qubits.
	}\label{fig:multi-controlled-x-half}
\end{figure}

We also consider a more structured input circuit, namely implementing half of a multi-controlled-$X$ gate, \cknot{k}.
We consider decompositions of \cknot{k} for $k$ integer powers of $2$,
but only compile the first half of the circuit, given in \cref{fig:multi-controlled-x-half-circuit}.
A $T$-efficient implementation of \cknot{k} uses measurement and feedback for uncomputation~\cite{Jones2013},
which are not captured in our model (see \cref{sec:conclusion}).
We plot the space-time cost of compiling the half \cknot{k} in \cref{fig:multi-controlled-x-half-results}.
We see that the dependence on the number of qubits $k$ is worse for \sw{}-based compilation,
and results in a larger space-time cost starting at 64 qubits.
Unfortunately, the \sw{}-based compilation is quite slow:
we ran the algorithm for at most 3 days and 9 hours at each data point
and were only able to obtain results up to 128 qubits.
However, the data we were able to obtain indicates a cross-over for compiling \cknot{k} circuits. The \sw{}-based compilation has better space-time performance for small instances, while EDPC has a better space-time performance for compiling large \cknot{k} circuits.

\section{Conclusion}\label{sec:conclusion}
In this paper, we have introduced the EDPC algorithm for the compilation of input quantum circuits into operations which can be implemented fault-tolerantly with the surface code.
The heart of this algorithm lies in the EDP subroutine,
which can implement both sets of parallel long-range \cnot{} gates
and sets of parallel rotations in constant depth
using existing efficient graph algorithms to find sets of edge-disjoint paths.
EDPC has advantages over other compilation approaches including Pauli-based computation, network coding based compilation, and \sw{}-based compilation.
We numerically find that EDPC significantly outperforms \sw{}-based circuit compilation in the space-time cost of random \cnot{} circuits for a broad range of instances, and for larger \cknot{k} gates.
However, many details of EDPC can be improved,
as it is only a first step towards using long-range operations for surface code compilation.

EDPC requires sets of constrained edge-disjoint paths,
which we call operator paths and run almost entirely along ancilla qubits.
Better algorithms for finding maximum sets of edge-disjoint operator paths could improve EDPC.
It seems likely that an $\bigo{\log n}$-approximation algorithm for finding maximum EDP sets on grids~\cite{Aumann1995} can be modified
to give an algorithm for finding maximum sets of edge-disjoint operator paths on grids.
A polylogarithmic approximation algorithm for this task
would imply an approximation algorithm for minimizing the depth, up to a polylogarithmic factor,
of compiling parallel \cnot{}s using the EDP subroutine.
In practice, it is, however, also important to find approximation algorithms with reasonable constant prefactors.

The runtime complexity of EDPC for an input circuit of depth $D$ acting on $n$ qubits is $\bigo{D n^3 \log n}$.
This is significantly faster than the \sw{}-based compilation in \cref{sec:compilation-swap}, which was found to be $\bigo{D n^5}$ in \textcite{Childs2019}.
We found that our implementation of the \sw{}-based compilation implementation runtime is much slower than that of EDPC on small instances, and found that the \sw{}-based algorithm had impractically long runtimes when applied to circuits beyond a few hundred qubits,
the regime of large-scale applications of quantum algorithms~\cite{Reiher2017,Gidney2021}.
Potential ways to further improve EDPC's runtime
include using a dynamical decremental all-pair shortest path algorithm in the greedy
approximation of the maximum EDP set,
or by finding faster and better approximation algorithms for finding the maximum set of edge-disjoint operator paths.

Any diagonal gates in the $Z$ (or $X$) basis can be performed remotely on the boundary,
including CCZ gates~\cite{Gidney2019} (see \cref{sec:RemoteGate}).
Therefore, our results on applying $Z(\theta)$ rotations can be extended to diagonal gates,
which will benefit circuit depth.

Even with the capability to perform long-range operations it may still be helpful to localize the quantum information
on some part of the architecture
such as by permuting the data qubits.
In particular, the size of the EDP set is bounded above by the minimum edge cut separating the terminals.
Therefore, it may be beneficial to first redistribute quantum information where it is needed
to ensure large EDP solutions exist.
It is straightforward to construct a long-range move of a data qubit to an ancilla in depth 2
from a long-range \cnot{},
by performing the \cnot{} targeting a $\ket 0$ ancilla state and measuring the source in the $X$ basis up to Pauli corrections.
It also is straightforward to adapt the EDP subroutine to perform sets of these long-range moves along operator paths,
now ending at the ancilla, in depth 4.
The depth to permute only a few qubits a long distance can be improved significantly by this technique.
For example, a \sw{} of the two corners of an $L\times L$ grid architecture takes $\bigo{1}$ depth using long-range moves,
as opposed to $\bigomega{L}$ depth using conventional \sw{}s.
It remains an open question how to trade off permuting data qubits (using \sw{}s or long-range moves)
and directly using long-range \cnot{}s.

We have assumed that classical feedback is not present in the input circuit for clarity of presentation.
EDPC can readily be extended to the setting of classical feedback in the input circuit to form a ``just-in-time'' surface code compilation algorithm.
To do so, a larger computation would be broken up into a sequence of circuit executions without classical feedback,
where prior measurement results specify the next circuit to compile and execute.

\subsection*{Acknowledgements}
We would like to thank Alexander Vaschillo and Dmitry Vasilevsky for helpful discussion,
and Thomas Haner for coming up with a simple correctness proof for delayed remote execution gadget
and many helpful discussions.
A slightly modified version of this proof is presented in \cref{sec:RemoteGate}.

E.S.\ acknowledges support by a Microsoft internship, an IBM PhD Fellowship,
and the U.S.\ Department of Energy, Office of Science, Office of Advanced Scientific Computing Research,
Quantum Testbed Pathfinder program (award number DE-SC0019040).

\printbibliography%

\appendix
\section{Surface code architecture}\label{sec:surfaceCode}

Here we review some basic details of the surface code focusing on the elementary logical operations shown in \cref{fig:logical-layout}.
This is intended as a high-level overview to provide some intuition of how the logical operations in \cref{fig:logical-layout} arise and what their resource costs are.
For more thorough reviews of surface codes see Refs.~\cite{Bombin2013,Fowler2012,Brown2017}.

To implement the surface code, we assume \textit{physical qubits} are laid out on the vertices of a 2D grid, with nearest-neighbor interactions allowed. 
For concreteness, we will describe here an implementation of lattice surgery with the rotated surface code with half-moon boundary~\cite{Landahl2014}, although our EDPC algorithm can use other implementations.
A single surface code patch encodes a single \emph{logical qubit} in $2d^2-1$ physical qubits, where the odd parameter $d$ is known as the \emph{code distance} which corresponds to the level of noise protection; see \cref{fig:surfaceCodeBasics}(a). 
For clarity, within this section of the appendix we refer to physical qubits and logical qubits explicitly, however in other sections we often drop the word ``logical'' when referring to logical qubits for brevity. 

\begin{figure}
	\centering
	\begin{subfigure}[t]{0.5\textwidth}
		\centering
		\includegraphics[width=0.75\textwidth]{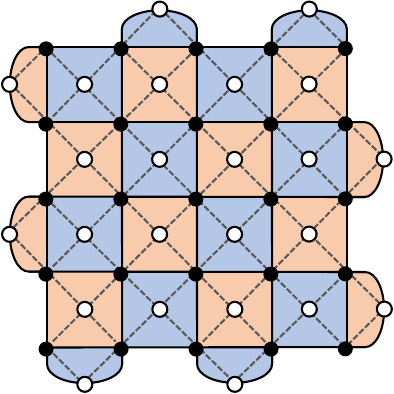}
		\caption{Surface code patch}\label{fig:surface-code-patch}
	\end{subfigure}%
	\begin{subfigure}[t]{0.5\textwidth}
		\centering
		\includegraphics[width=0.75\textwidth]{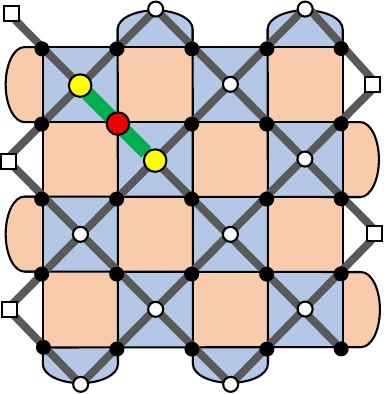}
		\caption{$X$-Stabilizer decoding graph}\label{fig:surface-code-decoding-graph}
	\end{subfigure}%
	\caption{%
		(a) A $d=5$ surface code patch implemented in a grid of data physical qubits (black disks), and ancilla physical qubits (white disks). 
		Error correction is implemented with single-qubit operations and \cnot{} between pairs of qubits connected by a dashed edge. 
		$Z$ and $X$ type stabilizers are associated with alternating red and blue faces.
		(b) A decoding graph that is defined by associating an edge with each qubit and a vertex for each stabilizer.
		If stabilizers are measured perfectly, $Z$ errors on data qubits (marked in red) can be corrected by finding a minimum weight matching (green edges) of vertices associated with unsatisfied $X$ stabilizers (yellow disks).
		}\label{fig:surfaceCodeBasics}
\end{figure}

We designate every odd physical qubit as a \emph{data physical qubit} in the patch, and every even physical qubit as an \emph{ancilla physical qubit} to facilitate a \emph{stabilizer} measurement; see \cref{fig:surfaceCodeBasics}(a).
The code space of a surface code consists of those states of the data physical qubits which are simultaneous $+1$ eigenstates of the set of stabilizer generators.
The stabilizer generators can be associated with faces and are either $X \otimes X \otimes X \otimes X$ or $Z \otimes Z \otimes Z \otimes Z$ operators for the bulk (interior) of the code or $X \otimes X$ or $Z \otimes Z$ operators on the boundary.
We can see that the logical $Z$ operator, $Z_L$,
defined as any path of single-qubit $Z$ operators on physical qubits connecting the rough boundaries,
commutes with all stabilizers.
Similarly, the logical $X$ operator, $X_L$, is a path of $X$ operators connecting the smooth boundaries.

For quantum error correction, it is necessary to repeatedly measure stabilizer generators.
Stabilizer generators can be measured by running small circuits consisting of the preparation of the ancilla physical qubit, \cnot{}s between the ancilla physical qubit and the data physical qubits, followed by measurement of the ancilla physical qubit.
Error correction can be performed by associating qubits with edges and stabilizer generators with vertices of a so-called decoding graph; see \cref{fig:surface-code-decoding-graph}. 
A classical algorithm known as a decoder is used to infer a set of edges (specifying the support of the X or Z correction) given a subset of vertices (corresponding to unsatisfied Z or X stabilizers, that is stabilizer generators with measurement outcome -1).
\cref{fig:surface-code-decoding-graph} shows an example of this in the setting of perfect stabilizer measurements, although this can be generalized to handle faulty measurements by repeating measurements.

\begin{figure}
	\centering
	\begin{subfigure}[t]{0.4\textwidth}
		\centering
		\includegraphics[width=0.6\textwidth]{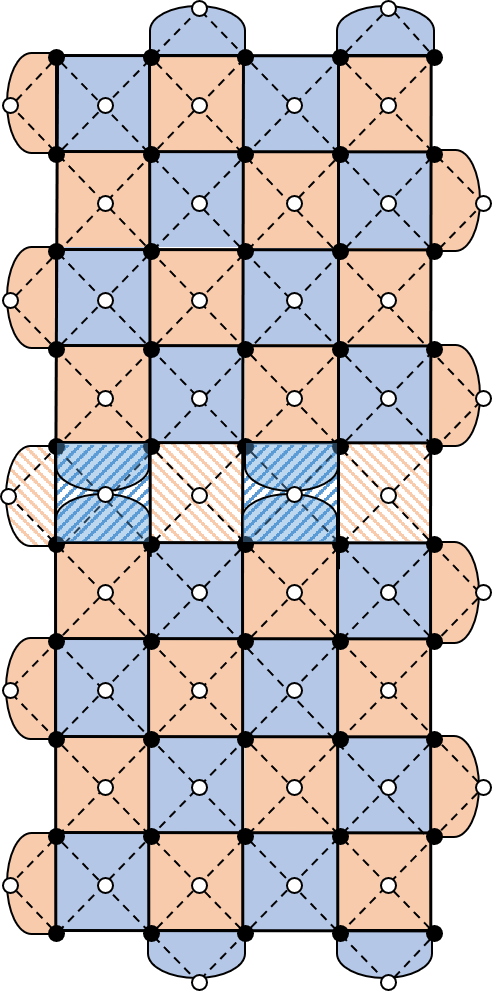}
		\caption{Logical $Z_L \otimes Z_L$ measurement}\label{fig:surface-code-merge}
	\end{subfigure}%
	\begin{subfigure}[t]{0.6\textwidth}
		\centering
		\includegraphics[width=0.8\textwidth]{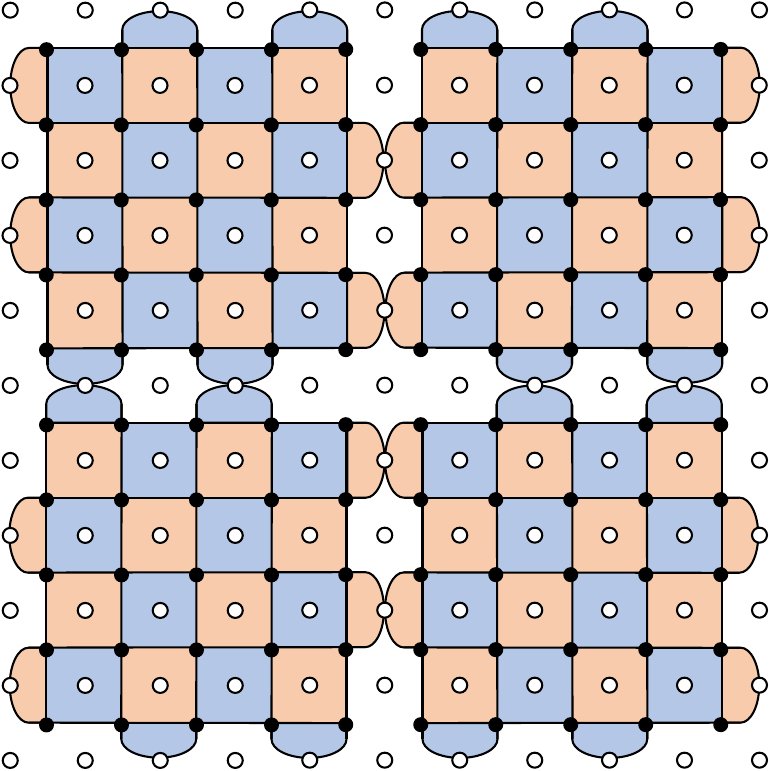}
		\caption{Patches tiling the plane}\label{fig:surface-code-plane}
	\end{subfigure}
	\caption{%
		(a) A logical $Z_L \otimes Z_L$ measurement is performed by lattice surgery in the following steps: (i) Stop measuring the weight-two stabilizers along the horizontal boundary between the patches. 
		(ii) Reliably measure the bulk faces for a single vertically-extended patch.
		Note that $Z_L\otimes Z_L$ can be inferred from the product of the outcomes of the newly measured red faces. 
		This temporarily merges the patches to form a single extended surface code patch.
		(iii) Reliably measure once more the weight-two faces along the horizontal boundary between the patches.
		This separates the pair of patches.
		(b) Two types of patches tile the plane, with $Z_L\otimes Z_L$ measurements possible between vertically neighboring patches, and $X_L\otimes X_L$ measurements possible between horizontally neighboring patches.
		}\label{fig:surfaceCodeOperations}
\end{figure}

Logical operations can be implemented fault-tolerantly on logical qubits encoded in surface codes. 
For example, a destructive logical $X$ measurement of a patch is implemented by measuring all data qubits in the $X$ basis, and then using a decoder to process the physical outcomes and reliably identify the logical measurement outcome.
Another important logical operation is the non-destructive measurement of a logical joint Pauli operator using an approach known as lattice surgery \cite{Horsman2012} as shown in \cref{fig:surface-code-merge}.
To simplify lattice surgery by lining up the boundary stabilizers of neighboring patches, we consider a tiling of the plane using two versions of distance $d$ surface code patches as shown in \cref{fig:surface-code-plane} which forms a grid of logical qubits. 
Logical $Z_L \otimes Z_L$ can be measured between vertical neighbor patches while $X_L \otimes X_L$ can be measured between horizontal neighbor patches.


\begin{figure}
    \centering
    \includegraphics[scale=0.7]{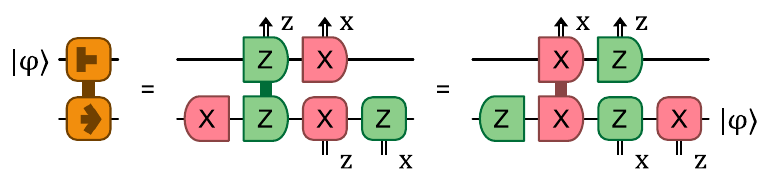}
    \caption{%
        The move operation can be implemented in depth 1 by local and neighboring Pauli measurements.
        A horizontal move can be implemented by preparing a single-qubit patch in $\ket{0}$,
        applying joint $XX$ measurement,
        and then measuring the original patch in the $Z$ basis (up to Pauli corrections).
        The vertical move follows from applying a Hadamard to the source qubit $\ket \phi$
        and a Hadamard on the output.
        Simplifying the circuit gives the right-hand side in the Figure,
        with a $ZZ$ measurement that is available vertically.
	}\label{fig:move}
\end{figure}

\begin{figure}
	\centering
	\begin{subfigure}{0.3\textwidth}
		\centering
		\includegraphics[width=\textwidth]{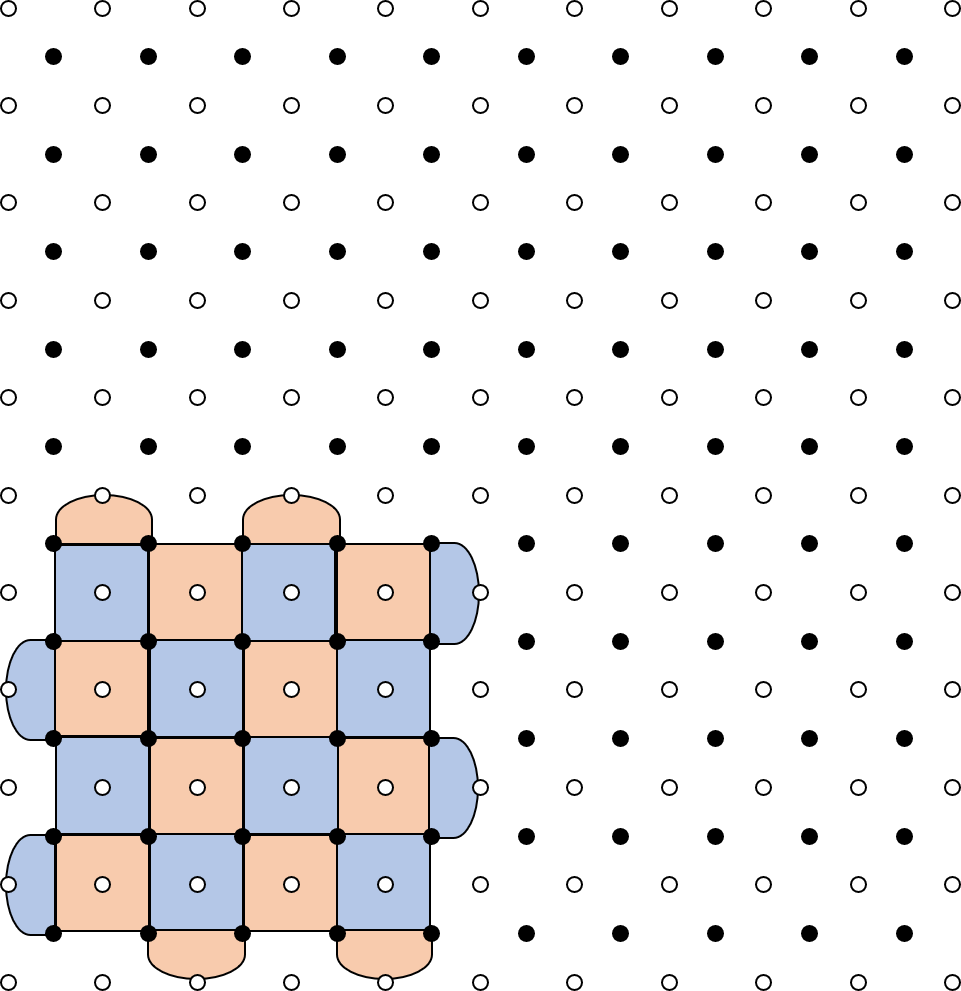}
		\caption{Transversal Hadamard}\label{fig:hadamard-step1}
	\end{subfigure}
	\hfill
	\begin{subfigure}{0.3\textwidth}
		\centering
		\includegraphics[width=\textwidth]{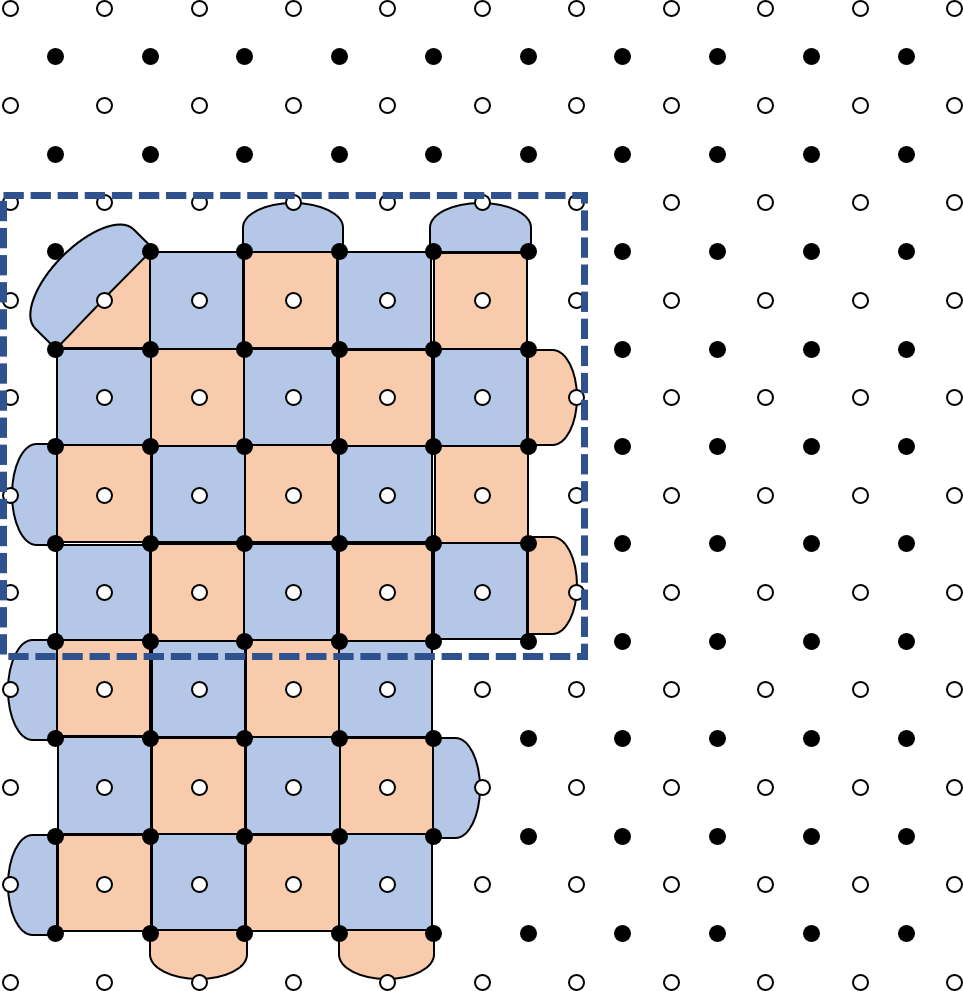}
		\caption{Extended patch}\label{fig:hadamard-step2}
	\end{subfigure}
	\hfill
	\begin{subfigure}{0.3\textwidth}
		\centering
		\includegraphics[width=\textwidth]{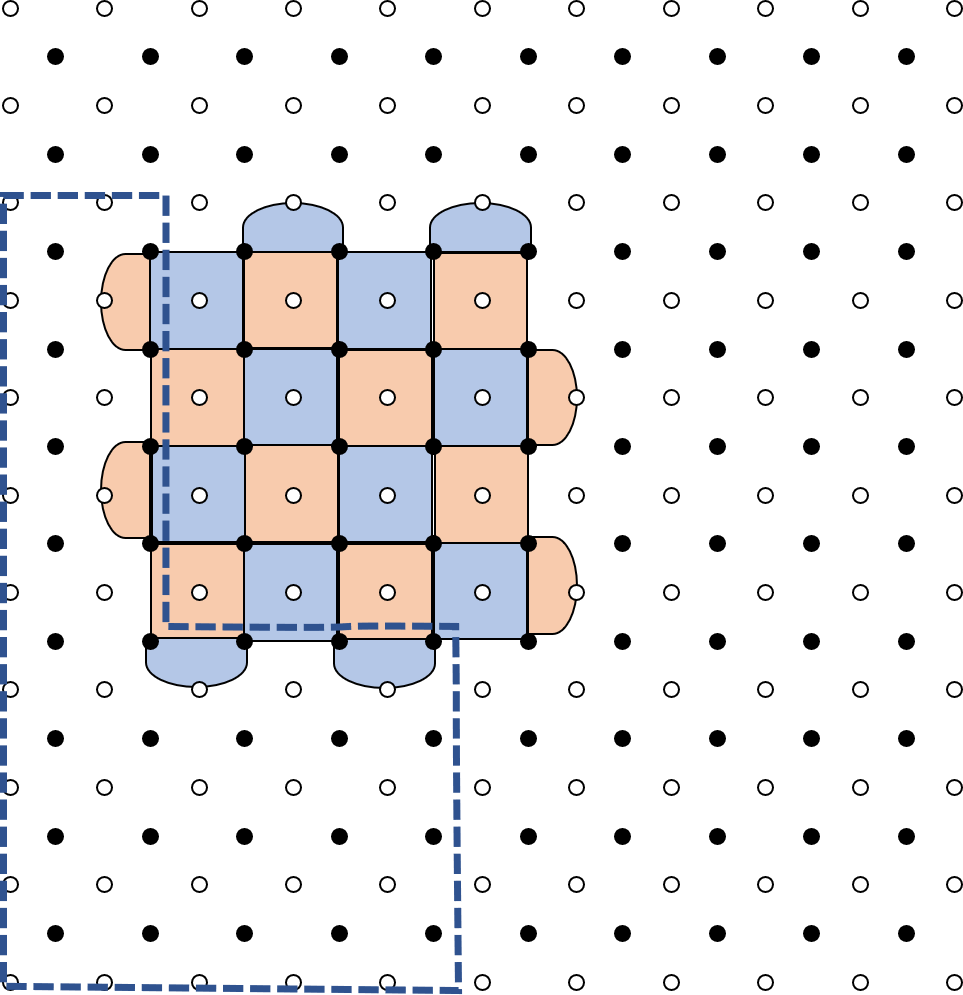}
		\caption{Shrunk patch}\label{fig:hadamard-step3}
	\end{subfigure}
	\caption{%
		Implementation of a Hadamard operation in depth 3 with three ancilla patches.
		(\subref{fig:hadamard-step1}) A transverse Hadamard is applied in depth 0 to each physical data qubit,
		which switches the arrangement of $X$ and $Z$ stabilizer generators compared to the standard configuration.
		(\subref{fig:hadamard-step2}) The patch is extended in depth 1 so that a segment of the standard boundary type 
		is introduced on the right.
		(\subref{fig:hadamard-step3}) The patch is shrunk into a standard surface code patch of the form of the top-left corner of the region (see \cref{fig:surface-code-plane}) in depth 1,
		but with its location shifted by a (code distance) $d$-independent amount.
		This allows us to shift the patch into the top-left corner in 0 depth (not shown).
		Then we move the logical qubit to the bottom-left corner in depth 1.
	}\label{fig:hadamard}
\end{figure}

\begin{figure}
        \centering
        \begin{subfigure}{\textwidth}
        	\centering
        	\includegraphics[scale=0.6]{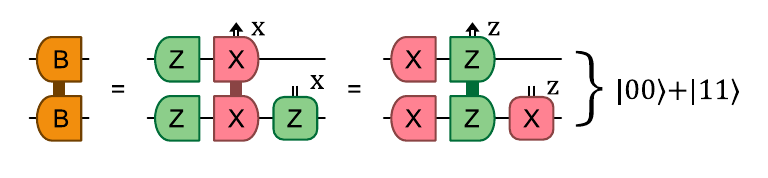}
        	\caption{Bell pair preparation}\label{fig:bell-operations-preparation}
        \end{subfigure}
        \begin{subfigure}{\textwidth}
        	\centering
        	\includegraphics[scale=0.6]{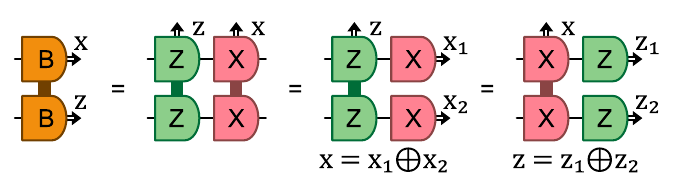}
        	\caption{Bell measurement}\label{fig:bell-operations-measurement}
        \end{subfigure}
    \caption{%
		We can implement Bell preparation and measurement in terms of single and two-qubit Pauli measurements
		in depth 1
		as given in \cref{fig:bell-operations}~\cite{Litinski2019}.
		(\subref{fig:bell-operations-preparation}) A Bell pair can be prepare from a (horizontal) joint $XX$ measurement of $\ket{00}$
		or a (vertical) joint $ZZ$ measurement of $\ket{++}$,
		up to Pauli corrections.
		(\subref{fig:bell-operations-measurement}) A destructive Bell measurement can be implemented
		by a joint $XX$ measurement followed by individual $Z$ basis measurements,
		or by a joint $ZZ$ measurement followed by individual $X$ basis measurements.
    }\label{fig:bell-operations}
\end{figure}

The allowed fault-tolerant logical operations that we assume throughout the paper and the resources they require are listed in \cref{fig:logical-layout}.
These are largely based on the rules specified in \cite{Litinski2019}.
Here we justify the resource requirements for the logical operations  in \cref{fig:logical-layout}
not covered in~\cite{Litinski2019}
on a distance-$d$ surface code.
For space analysis, we work in units of full surface code patches such that if any qubits from a patch are needed to implement an operation the full patch is counted.
We show how to implement the operations in terms of more elementary Pauli measurements.
The move operation can be implemented in depth 1 with the target qubit as ancilla, as shown in \cref{fig:move}.
The Hadamard can be implemented in depth three with three ancilla patched
along with the move operation as shown in \cref{fig:hadamard}.
Finally,
Bell measurement and preparation can be implemented in depth 1 as shown in \cref{fig:bell-operations}.

It is worth mentioning that there is considerable freedom in the detailed choice and implementation of the surface code which could have an impact on the space-time cost of logical operations, both at the physical level but also in some cases at the logical level.
For example the Hadamard could be performed using just one logical ancilla patch if each patch was padded with extra qubits.
We do not explore these alternatives here, but note that our EDPC algorithm can still be applied if these alternatives are used.

\section{Logical space time cost as a proxy for physical space time cost}%
\label{sec:logical-physical-cost}

Here we provide a justification for our use of logical space time cost as a proxy for physical space time cost.
As we have seen in \cref{fig:logical-layout} and \cref{sec:surfaceCode},
logical operations implemented with the surface code require physical time that scales as $d$ and physical space that scales as $d^2$. 
For a logical circuit written in terms of a total of $A_{\text{logical}}$ elementary logical operations implemented using surface codes of distance $d$,
the physical space-time cost $A_{\text{physical}}$ is approximately
\begin{equation}
A_{\text{physical}} \sim A_{\text{logical}} d^3.
\end{equation}
The probability of any of these elementary operations resulting in a logical failure scales as $p_{\text{fail}} \sim (p/p^*)^{d/2}$, where the fixed system parameters are the physical error rate $p$, and the fault-tolerant threshold for the surface code $p^*$.
Moreover, we assume $p_{\text{fail}} \sim 1/A_{\text{logical}}$ to ensure that the logical circuit is reliable with as small a code distance as possible. 
This suggests that the code distance behaves as
\begin{equation}
	d \sim \frac{2\log A_{\text{logical}}}{\log p^*-\log p}.
\end{equation}
Therefore we see that the physical and logical space time costs are monotonically related, i.e.,
\begin{equation}
A_{\text{physical}} \sim A_{\text{logical}} (\log A_{\text{logical}})^3.
\end{equation}

\section{\cnot{} via Bell operations}\label{sec:cnots-via-bell}
We list more variations of the standard \cnot{} gate (\cref{fig:cnotGate})
that use intermediate Bell preparation and measurements on ancillas
in \cref{fig:cnot-bell-variations}.
By choosing the right subcircuit,
we see that the long-range operations in \cref{fig:long-range-cnot} implement a \cnot{} gate.

\begin{figure}[htbp]
	\centering
	\begin{subfigure}[t]{0.5\textwidth}
		\centering
		\includegraphics[scale=0.6]{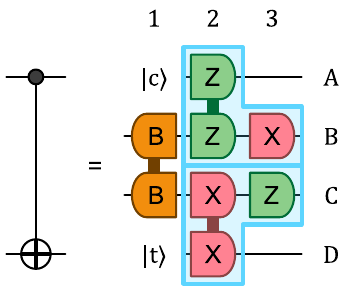}
		\caption{\cnot{} via Bell preparation}
		\label{fig:cnot-bell}
	\end{subfigure}%
	\begin{subfigure}[t]{0.5\textwidth}
		\centering
		\includegraphics[scale=0.6]{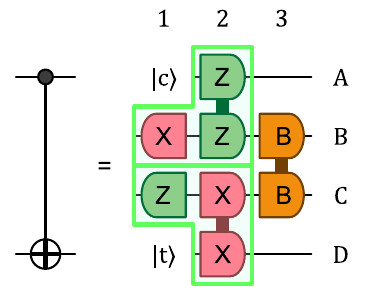}
		\caption{\cnot{} via Bell measure}
	\end{subfigure}
	\begin{subfigure}[t]{0.5\textwidth}
		\centering
		\includegraphics[scale=0.6]{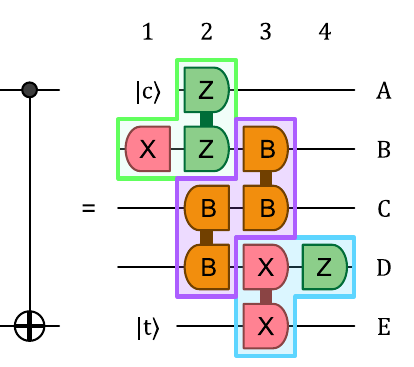}
		\caption{\cnot{} with control first}
	\end{subfigure}%
	\begin{subfigure}[t]{0.5\textwidth}
		\centering
		\includegraphics[scale=0.6]{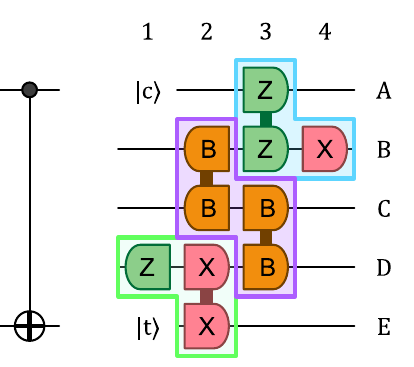}
		\caption{\cnot{} with target first}
	\end{subfigure}
	\caption{%
		Various implementations of a \cnot{} gate with intermediate ancilla qubits and Bell operations.
		In particular, we are able to apply the control and the target
		either before (green) or after (teal) Bell preparation and measurement steps,
		while keeping the depth at 2.
	}\label{fig:cnot-bell-variations}
\end{figure}

\section{Remote execution of diagonal gates}\label{sec:RemoteGate}
\begin{figure}
    \centering
    \begin{subfigure}[b]{0.48\textwidth}
    	\centering
    	\includegraphics[width=\textwidth]{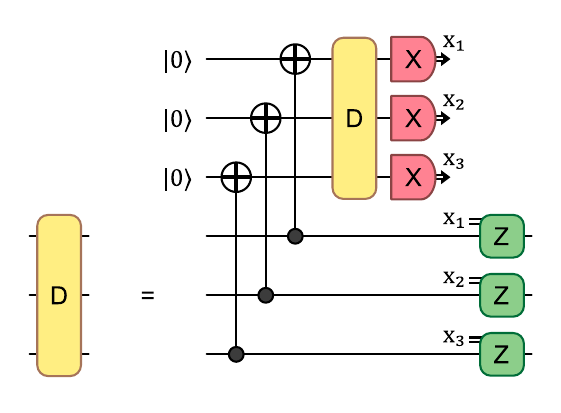}
    	\caption{Diagonal gate in computational basis}\label{fig:remote-diagonal-computational}
    \end{subfigure}\hfill
    \begin{subfigure}[b]{0.48\textwidth}
    	\centering
    	\includegraphics[width=\textwidth]{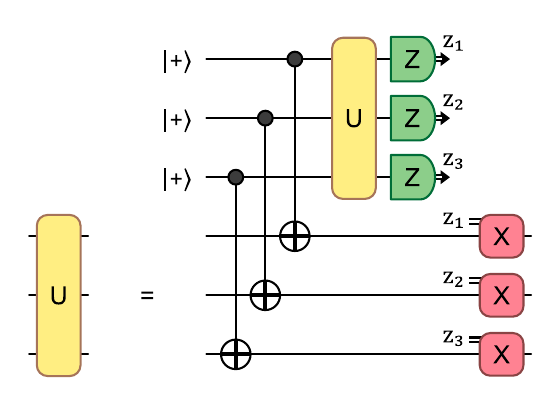}
    	\caption{Diagonal gate in Hadamard basis}\label{fig:remote-diagonal-hadamard}
    \end{subfigure}
	\caption{%
		(\subref{fig:remote-diagonal-computational}) Any $k$-qubit gate diagonal in the computational basis
		can be remotely executed on $k$ dedicated ancilla by first using \cnot{}s.
		We use this technique to apply remote $Z(\theta)$ rotations (\cref{fig:delayed-remote-gate})
		with magic states at the boundary.
		(\subref{fig:remote-diagonal-hadamard}) Similarly, gates diagonal in the Hadamard basis
		also have a remote implementation.
		Since the Pauli corrections can be commuted through Clifford circuits,
		Clifford circuits can be executed immediately after executing the \cnot{} operations
		with no need to wait on the remote operations.
	}\label{fig:remote-diagonal}
\end{figure}

A gate $D$ diagonal on $k$ source qubits in the computational basis can be executed 
on $k$ ancilla by first entangling these ancilla qubits using \cnot{}s.
We call this \emph{remote} execution.
Let the computational basis be $\ket{\ell}$, for $\ell \in [2^k]$,
then $D\ket{\ell} = \exp(i\phi_\ell) \ket{\ell}$.
We saw one use for remote gates
in applying rotations at the boundary requiring magic states (\cref{sec:magicState}).

We execute $D$ remotely as follows (see \cref{fig:remote-diagonal}).
First, we initialize the ancilla in the state $\ket{0}^{\otimes k}$.
Let the source qubits be in some pure state $\sum_\ell \alpha_\ell \ket{\ell}$,
for $\alpha_\ell \in \mathbb C$.
then we apply $k$ transversal \cnot{}
gates controlled on source qubits so that the overal state becomes $\sum_\ell \alpha_\ell \ket{\ell}\otimes\ket{\ell}$.
We now apply $D$ to the ancilla instead
\begin{equation}
	(\idm \otimes D)\sum_\ell \alpha_\ell \ket{\ell}\otimes\ket{\ell} =
	\sum_\ell \alpha_\ell \exp(i\phi_\ell)\ket{\ell}\otimes\ket{\ell}.
\end{equation}
We now disentangle the ancilla by measuring them in the $X$ basis.
Let the measurement give outcomes $\vec x \in \set{0,1}^k$,
then the state on the source qubits is mapped to
\begin{equation}
	\sum_\ell \alpha_\ell \exp(i\phi_\ell) (-1)^{(\vec x,\ell)}\ket{\ell},
\end{equation}
where $(\vec x,\ell)$ is the inner product modulo 2 between $\vec x$
and the binary representation of $\ell$.
Applying a $Z$ correction to each qubit $j \in [k]$ controlled on measurement result $\vec x_j$
maps the state to $\sum_\ell \alpha_\ell  \exp(i\phi_\ell) \ket{\ell}$ as required.

This technique can be extended to any unitary operator $U$ 
since it can be unitarily diagonalized as $U = VDV^\dagger$ by the spectral theorem,
for $V$ unitary and $D$ diagonal operators.
A particularly simple case are unitary operators that are diagonal in the Hadamard basis,
where $V = H^{\otimes k}$.
We write $U = H^{\otimes k} D H^{\otimes k}$ on the source qubits
and apply remote execution of $D$ using our techniques above.
We then simplify the circuit to obtain \cref{fig:remote-diagonal-hadamard}.

\section{EDPC implementation}\label{app:edpcImplementation}
Here we provide \cref{alg:edpcImplementation}, which specifies the implementation of EDPC used for our numerical results presented in \cref{sec:numerical-results},
here called EDPCI for clarity.
EDPCI differs slightly from EDPC (\cref{sec:EDPCalogrithm}) and we will highlight the differences.
Up until \cref{line:impl-rotations} in \cref{alg:edpcImplementation}, EDPCI is the same as EDPC.
Then, EDPCI greedily attempts to execute long-range \cnot{}s earlier than would occur in EDPC.
In particular, EDPC only executes \cnot{}s after all available rotations have been executed,
whereas EDPCI finds a set $\mathcal P_c$ on \cref{line:impl-cnot}
such that $\mathcal P_c \cup \mathcal P_m$ forms an EDP set.
Now EDPCI concurrently executes long-range \cnot{}s using any edges left over from remote rotations.
Moreover, we note that EDPC uses the bounded $\mathcal T$-operator set algorithm (\cref{alg:boundedTOperator})
to execute a parallel \cnot{} circuit,
which additionally finds a bounded $\mathcal T$-operator set $\mathcal Q_1$ on \cref{line:edpc-q1},
whereas EDPCI only finds $\mathcal Q_2$ from the bounded $\mathcal T$-operator set algorithm if given a parallel \cnot{} circuit.
As a consequence of this difference, while a parallel input \cnot{} circuit is guaranteed to compile to an output circuit upper bounded by $\bigo{\sqrt n}$, EDPCI does not have this guarantee.

\begin{algorithm}[htbp]
	\caption{%
		\emph{EDPC implementation}:
		The EDPC algorithm (\cref{alg:edpc}) differs from our implementation in that it greedily tries to execute \cnot{}s earlier.
	}\label{alg:edpcImplementation}
	\Input{Circuit $\mathcal C$ with Paulis commuted to the end and merged with measurement}

	\While{available operations in $\mathcal C$}{%
		\execute all available state preparation, measurement, and Hadamard\;

		$\mathcal G_m \gets \set{\text{available operations } g \in \mathcal C \mid g \text{ is } S,T, S_x,$ \text{or} $T_x}$\;
		$\mathcal G_c \gets \set{\text{available operations } g \in \mathcal C \mid g \text{ is \cnot{}}}$\;
		\While{$\mathcal G_m \cup \mathcal G_c \neq \emptyset$\label{alg:longRangeWhile}}{%
			$G \gets$ surface code grid graph\;
			$\mathcal P_m \gets$ \KwMaxRotations{$\mathcal G_m$}\tcp*{see \cref{alg:magic-rotations}}\label{line:impl-rotations}
			remove edges in each $P \in \mathcal P_m$ from $G$\;

			$\mathcal P_c \gets$ approximate max operator EDP set on $G$ with $\mathcal G_c$\label{line:impl-cnot}\;

			\execute concurrent remote rotations along $\mathcal P_m$ and long-range \cnot{}s along $\mathcal P_c$ using EDP subroutine\;
			remove executed rotations from $\mathcal G_m$ and \cnot{}s from $\mathcal G_c$\;
		}
	}
\end{algorithm}

\end{document}